\documentclass[letterpaper, 10 pt, conference]{ieeeconf}  % Comment this line out
\usepackage{graphicx}
\usepackage{amsmath}
\usepackage{mathtools}
\usepackage{amssymb}
\usepackage{amsfonts}
\usepackage{upgreek}
\usepackage{hyperref}
\usepackage{algorithm}
\usepackage{algorithmic}
\usepackage{mathdots}
\usepackage{bm}
\usepackage{textcomp}
\usepackage{cases}
\usepackage{caption}
\usepackage{subcaption}
\usepackage{color}
\usepackage{soul}
\usepackage{tikz} % Package for drawing
\usetikzlibrary{matrix, calc, arrows, shapes, positioning, angles}
\usepackage{amsmath}
\usepackage{hyperref}
\usepackage{todonotes}
\usepackage{wasysym}
\usepackage{soul}
\usepackage[normalem]{ulem}
\usepackage{tikz}
\usepackage{pgfplots}
\usetikzlibrary{intersections}
\usetikzlibrary{patterns}
\usepgfplotslibrary{fillbetween}
\usetikzlibrary{arrows.meta}
\usepackage{filecontents}
\usepackage{comment}
\usepackage{cite}
\usepackage{cleveref}
\usepackage{optidef}
\usepackage[most]{tcolorbox} % Import tcolorbox package with the 'most' library for additional features
\usetikzlibrary{arrows.meta, positioning, shapes.geometric, shadows, decorations.pathmorphing, decorations.pathreplacing, fit}

\newtheorem{definition}{Definition}
\newtheorem{lemma}{Lemma}
\newtheorem{remark}{Remark}
\newtheorem{theorem}{Theorem}
\newtheorem{assumption}{Assumption}

\newtheorem{example}{Example}

\newcommand{\bL}{\bm{L}}

\newcommand{\ba}{\bm{a}}
\newcommand{\bx}{\bm{x}}
\newcommand{\bxs}{\bm{x}^{\ast}}
\newcommand{\bxsh}{\hat{\bm{x}}^{\ast}}

\newcommand{\by}{\bm{y}}

\newcommand{\bu}{\bm{u}}

\newcommand{\bA}{\bm{A}}
\newcommand{\bz}{\bm{z}}

\newcommand{\mbe}{\mathbf{e}}
\newcommand{\mbI}{\mathbf{I}}

\newcommand{\V}{\mathcal{V}}

\newcommand{\D}{\mathcal{D}}
\newcommand{\Gtask}{\mathcal{G}^{\psi}}
\newcommand{\Gcom}{\mathcal{G}^c}
\newcommand{\Etask}{\mathcal{E}^{\psi}}
\newcommand{\Ecom}{\mathcal{E}^c}
\newcommand{\Ntask}{\mathcal{N}^{\psi}}
\newcommand{\Ncom}{\mathcal{N}^c}

\newcommand{\R}{\mathbb{R}}

\newcommand{\alphb}{\bar{\alpha} }
\newcommand{\betastr}{\beta^\ast}

\newcommand{\I}{\mathcal{I}}

\newcommand{\bxtil}{\tilde{\bm{x}}}
\newcommand{\bxtils}{\tilde{\bm{x}}^\ast}

\newcommand{\bxtildot}{\dot{\tilde{\bm{x}}}}

\newcommand{\nualp}{\nu_{\alpha}}
\newcommand{\nubet}{\nu_{\beta}}

\newcommand{\mathdot}{\mathord{\cdot}}

\newcommand{\col}{\mathrm{col}}

\newcommand{\xest}[2]{\bx_{#1}^{(#2)}}
\newcommand{\xestvec}[1]{\bx^{(#1)}}
\newcommand{\xestvecdot}[1]{\dot{\bx}^{(#1)}}
\newcommand{\mbN}{\mathbb{N}}

\newcommand{\betab}{\bar{\beta}}
\newcommand{\betabstr}{\bar{\beta}^\ast}

\newcommand{\bxc}[1]{\bm{x}_{\I_{#1}}}
\newcommand{\bxcs}[1]{\bm{x}^{\ast}_{\I_{#1}}}

\newcommand{\ftil}{\tilde{f}}
\newcommand{\Lapla}{\mathcal{L}}

\def\@IEEEtablestring{table}

\long\def\@makecaption#1#2{%
	\ifx\@captype\@IEEEtablestring%
	\begin{center}{\footnotesize #1}\\{\footnotesize\scshape #2}\end{center}%
	\@IEEEtablecaptionsepspace% V1.6 was a hard coded 8pt
	\else
	\@IEEEfigurecaptionsepspace% V1.6 was a hard coded 5pt
	\setbox\@tempboxa\hbox{\footnotesize #1.~~ #2}%
	\ifdim \wd\@tempboxa >\hsize%
	\setbox\@tempboxa\hbox{\footnotesize #1.~~ }%
	\parbox[t]{\hsize}{\footnotesize \noindent\unhbox\@tempboxa#2}%
	\else%
	\ifcenterfigcaptions \hbox to\hsize{\footnotesize\hfil\box\@tempboxa\hfil}%
	\else \hbox to\hsize{\footnotesize\box\@tempboxa\hfil}%
	\fi\fi\fi}

\IEEEoverridecommandlockouts                              % This command is only
\title{\LARGE \bf
Collaborative Satisfaction of Long-Term Spatial Constraints in Multi-Agent Systems: A Distributed Optimization Approach 
\\ (extended version)
}

\author{Farhad Mehdifar, Mani H. Dhullipalla, Charalampos P. Bechlioulis, and Dimos V. Dimarogonas% <-this % stops a space
\thanks{This work is supported by ERC CoG LEAFHOUND, the KAW foundation, and the Swedish Research Council (VR).}% <-this % stops a space
\thanks{F. Mehdifar, M. H. Dhullipalla, and D. V. Dimarogonas are with the Division of Decision and Control Systems, KTH Royal Institute of Technology, Stockholm, Sweden.   {\tt\small mehdifar@kth.se; manihd@kth.se; dimos@kth.se}}%
\thanks{C. P. Bechlioulis is with the Division of Systems and Control of the Department of Electrical and Computer Engineering at University of Patras, Patra, Greece. {\tt\small chmpechl@upatras.gr}}%
}

\begin{document}

\maketitle
\thispagestyle{empty}
\pagestyle{empty}

%%%%%%%%%%%%%%%%%%%%%%%%%%%%%%%%%%%%%%%%%%%%%%%%%%%%%%%%%%%%%%%%%%%%%%%%%%%%%%%%
\begin{abstract}
	This paper addresses the problem of collaboratively satisfying long-term spatial constraints in multi-agent systems. Each agent is subject to spatial constraints, expressed as inequalities, which may depend on the positions of other agents with whom they may or may not have direct communication. These constraints need to be satisfied asymptotically or after an unknown finite time. The agents' objective is to collectively achieve a formation that fulfills all constraints. The problem is initially framed as a centralized unconstrained optimization, where the solution yields the optimal configuration by maximizing an objective function that reflects the degree of constraint satisfaction. This function encourages collaboration, ensuring agents help each other meet their constraints while fulfilling their own. When the constraints are infeasible, agents converge to a least-violating solution. A distributed consensus-based optimization scheme is then introduced, which approximates the centralized solution, leading to the development of distributed controllers for single-integrator agents. Finally, simulations validate the effectiveness of the proposed approach.
\end{abstract}
%%%%%%%%%%%%%%%%%%%%%%%%%%%%%%%%%%%%%%%%%%%%%%%%%%%%%%%%%%%%%%%%%%%%%%%%%%%%%%%%
\section{Introduction}

Control and coordination of Multi-Agent Systems (MAS) have been a major research focus in the past decade, driven by tasks that require collaboration which are otherwise nearly impossible to achieve. Traditional MAS problems include consensus, rendezvous, flocking, formation, coverage, and containment \cite{chen2019control, oh2015survey, mehdifar20222}. However, recent research has shifted toward new demands, such as distributed optimal coordination \cite{zhang2017distributed} and handling high-level spatiotemporal (i.e., space and time) specifications in multi-robot systems \cite{lindemann2020barrier,lindemann2019feedback,liu2025controller,chen2024cooperative}, which do not explicitly lie within the classical MAS problems.

Distributed Optimization (DO) often involves minimizing a joint objective function using algorithms deployed across a network of communicating computation nodes (agents), where each agent knows only a small part of the problem and can communicate with a limited number of neighbors \cite{yang2019survey, nedic2018distributed, shorinwa2024distributed, wang2011control}. In multi-robot systems, DO offers a framework for developing local decision-making rules, addressing various challenges in cooperative robotics, such as surveillance, task allocation, optimal consensus, cooperative motion planning, self-organization, cooperative estimation, target tracking, and distributed SLAM. For a detailed exploration of DO's applications in multi-robot networks, see \cite{tron2016distributed, jaleel2020distributed, testa2023tutorial, shorinwa2024distributed}.

This paper introduces the problem of collaborative coordination in multi-agent systems under long-term spatial constraints. Each agent's position is subject to inequality type constraints that may depend on the positions of other agents, with whom they may or may not have direct communication. The objective is to collectively achieve a desired formation that satisfies all constraints. These constraints are long-term, as they only need to be satisfied asymptotically or after an unknown finite time. Furthermore, each agent must meet its own constraints while helping others satisfy theirs, despite lacking explicit knowledge of the other agents' constraints. 

Unlike distributed aggregative optimization in cooperative robotics \cite{li2021distributed, carnevale2022aggregative, testa2023tutorial}, where each agent’s local objective function must depend on the positions of all other agents, our formulation does not require this global dependency, thereby addressing a broader class of coordination problems.

We first reformulate the problem as a centralized optimization task, where maximizing the objective function results in a desired configuration that satisfies all constraints. The objective function is designed such that its positive values represent the satisfaction of the multi-agent constraints, with larger values indicating better overall constraints fulfillment. In cases where the constraints are collectively infeasible, solving the optimization problem provides a least-violating solution, reflected by a negative optimal value of the objective function. Next, we develop a novel multi-agent objective function as a sum of agents' local (private) objective functions, which depend solely on each agent’s spatial constraints. We demonstrate that minimizing the multi-agent objective function yields an approximate solution to the centralized optimization problem. This enables the design of distributed control protocols for single-integrator agents using distributed continuous-time consensus-based optimization algorithms. Additionally, we explore the sufficient conditions for convexity and strictly convexity of the multi-agent system's global objective function.

%Finally, note that although our solution approach closely mirrors typical Continuous Distributed Constraint Optimization Problems (C-DCOPs) in MAS \cite{fransman2023distributed,hoang2020new,fioretto2018distributed}, in our method each agent maintains a private objective function, whereas in conventional C-DCOP formulations, neighboring agents share objective functions along edges of a constraint graph.

%%%%%%%%%%%%%%%%%%%%%%%%%%%%%%%%%%%%%%%%%%%%%%%%%%%%%%%%%%%%%%%%%%%%%%%%%%%%%%%%
\section{Preliminaries}
\label{Sect:Prelim}

\noindent\textbf{Notation:} $\mbN$ and $\R$ denote the sets of natural and real numbers, respectively, while $\R^n$ represents the $n$-dimensional real space. Bold lowercase symbols denote vectors and vector functions, and bold uppercase symbols represent matrices. Non-bold symbols indicate scalar functions and variables. $\ba \in \R^n$ is an $n \times 1$ column vector, with $\ba^\top$ as its transpose. The Euclidean norm of $\ba$ is $\|\ba\|$. The concatenation operator is defined as $\col(\ba_i)_{i=1,\ldots,m} \coloneqq [\ba_1^\top, \ldots, \ba_m^\top]^\top \in \R^{mn}$, where $\ba_i \in \R^n$ for all $i \in \{1,\ldots,m\}$. The space of real $n \times m$ matrices is $\R^{n \times m}$. For a matrix $\bA \in \R^{n \times m}$, $\bA^\top$ denotes its transpose, and $\|\bA\|$ its induced norm. The absolute value of a real number is $|\cdot|$. $\mathbf{0}_n \in \R^n$ and $\mathbf{1}_n \in \R^n$ are the zero and ones vectors, respectively. $\mbI_n \in \R^{n \times n}$ is the $n$-dimensional identity matrix. The symbol $\otimes$ denotes the Kronecker product, and $\mbe_n^i \coloneqq [0 \ldots 1 \ldots 0]^\top \in \R^n$ is the $i$-th coordinate vector.

In the following we review the concept of log-convexity, which is a stronger notion of convexity.
\begin{definition}[Log-Convex Functions \cite{boyd2004convex, constantin2018convex}]
	A function $f : \R^n \rightarrow \R$ with convex domain $\D_f$ is log-convex if $f(\bx) > 0$ for all $\bx \in \D_f$ and $\log f(\bx)$ is convex. Equivalently, $f$ is log-convex if for all $\bx, \by \in \D_f$ and $0 \leq \theta \leq 1$, the following holds:
	\begin{equation} \label{eq:log_cvx_def}
		f(\theta \bx + (1-\theta)\by) \leq f(\bx)^\theta f(\by)^{1-\theta}.
	\end{equation}
	If strict inequality holds in \eqref{eq:log_cvx_def} for $\bx \neq \by$ and $0 < \theta < 1$, then $f$ is strictly log-convex, meaning $\log f$ is strictly convex.
\end{definition}

Log-convexity implies convexity, but the reverse does not hold. The following lemma summarizes operations that preserve log-convexity.
\begin{lemma}[Log-Convexity Preserving Operations] \label{lem:log_cvx_perserve}
	Log-convexity is preserved under positive scaling, positive powers, multiplication, and addition \cite{boyd2004convex,constantin2018convex}.
\end{lemma}

\begin{lemma}[Hölder's Inequality] \label{lem:holder_ineq}
	For $p>1$, $\frac{1}{p} + \frac{1}{q} = 1$, and $\bx = [x_1, \ldots, x_n]^\top,\by = [y_1, \ldots, y_n]^\top \in \R^n$ the following inequality holds:
	\begin{equation}
		\sum_{i=1}^{n} x_i y_i \leq \left(\sum_{i=1}^{n} |x_i|^p \right)^{\frac{1}{p}} \left(\sum_{i=1}^{n} |y_i|^q \right)^{\frac{1}{q}}.
	\end{equation}
\end{lemma}

%%%%%%%%%%%%%%%%%%%%%%%%%%%%%%%%%%%%%%%%%%%%%%%%%%%%%%%%%%%%%%%%%%%%%%%%%%%%%%%%
\section{Problem Formulation}
\label{Sect:ProblemFormu}

Consider a multi-agent system operating in a $d$-dimensional space, with $N \in \mbN$ agents, each governed by the single integrator dynamics:
\begin{equation}\label{eq:agent_single_integ_dyn}
	\dot{\bx}_i = \bu_i, \quad i = 1, \ldots, N,
\end{equation}
where $\bx_i \in \R^d$ represents the position of agent $i$, and $\bu_i \in \mathbb{R}^d$ is its velocity control input. Define the stacked vector of all agents' positions as $\bx \coloneqq [\bx_1^\top, \bx_2^\top, \ldots, \bx_N^\top]^\top \in \R^{Nd}$.

Let \( \V \coloneqq \{1, \ldots, N\} \) represent the index set of agents, and define \(\I_i \subseteq \V \setminus \{i\}\) as the set of agents with which agent \(i\) has spatial coupling constraints. Furthermore, let \(n_i \in \mathbb{N}\) denote the number of agents in \(\I_i\) (i.e., \( |\I_i| = n_i \)). Define \(\bxc{i} \coloneqq \col(\bx_j)_{j \in \I_i} \in \mathbb{R}^{n_i d}\) as the stacked positions of the agents that are involved in the spatial constraints of agent \(i\). Given this notation, assume agent \(i\) is subject to \(m_i \in \mathbb{N}\) spatial \textit{long-term constraints}, formulated as inequalities:
\begin{equation} \label{eq:agent_const_predicateform}
	\psi_{i,k}(\bx_i,\bxc{i}) > 0, \quad k = 1, \ldots, m_i, \quad i = 1 , \ldots, N,
\end{equation}
where \(\psi_{i,k}: \mathbb{R}^d \times \mathbb{R}^{n_i d} \rightarrow \mathbb{R}\) are continuously differentiable constraint (or task) functions. These constraints are termed long-term because they are required to hold only after an unknown but finite time, or as \(t \rightarrow +\infty\).

\begin{remark}\label{rem:agent i's_constraint_forms}
	The notation $\psi_{i,k}(\bx_i, \bxc{i})$  does not imply that all of agent $i$'s constraints depend on every other agent's position in the set $\I_i$. For instance, in a MAS with 5 agents, where $\I_1 = \{2,3,5\}$ and $m_1 = 4$, agent 1's constraint functions could take the form $\psi_{1,1}(\bx_1) > 0, \psi_{1,2}(\bx_1, \bx_2) > 0, \psi_{1,3}(\bx_1, \bx_5) > 0$, and $\psi_{1,4}(\bx_1, \bx_2, \bx_3) > 0$.
\end{remark}

\textbf{Objective:} Develop a distributed control protocol that guides agents to a desired spatial multi-agent formation, \( \bxs_1, \bxs_2, \ldots, \bxs_N \), satisfying all constraints \( \psi_{i,k}(\bxs_i, \bxcs{i}) > 0 \), for \( k = 1, \ldots, m_i \) and \( i = 1, \ldots, N \), assuming the constraints are feasible. If any constraints are infeasible, the protocol should instead direct the agents toward a formation in \( \mathbb{R}^d \) with the least violation of constraints.

%%%%%%%%%%%%%%%%%%%%%%%%%%%%%%%%%%%%%%%%%%%%%%%%%%%%%%%%%%%%%%%%%%%%%%%%%%%%%%%%
\section{Main Results}
\label{sec:main_results}

In this section, we recast the problem of collaboratively fulfilling spatial long-term constraints in MAS as an unconstrained distributed optimization problem. We first introduce and differentiate two graph types that model agents' coupling constraints (tasks) and their communication capabilities. Next, we reformulate the problem into an unconstrained centralized optimization framework and propose a distributed strategy to approximate its solution.

% -----------------------------------------------------------------------
\subsection{Task Dependency and Communication Graphs}
\label{sec:graphs}

Inspired by \cite{liu2025controller}, we introduce the task dependency graph for the MAS. As shown in \eqref{eq:agent_const_predicateform}, each agent’s constraints (tasks) may depend on the positions of other agents in the system. To capture this, we define a \textit{directed} graph \(\Gtask(\V, \Etask)\), referred to as the multi-agent \textit{task dependency graph}, where \(\Etask \subseteq \V \times \V\) represents the set of directed edges. An edge \((i,j) \in \Etask\) exists if any constraint function of agent \(i\) in \eqref{eq:agent_const_predicateform} depends on agent \(j\)’s position \(\bx_j\). A self-edge \((i,i)\) exists if agent \(i\) has a constraint that depends solely on \(\bx_i\), i.e., $\psi_{i,k}(\bx_i) > 0$. We define \(\Ntask_i = \{j \in \V \mid (i,j) \in \Etask \}\) as the set of agent \(i\)’s (out) neighbors in \(\Gtask\). The presence of an edge \((i,j)\) indicates that agent \(i\) must satisfy at least one constraint relative to its (out) neighbor agent \(j\), but not vice versa. A set \( \V^M \subseteq \V \) is called a \textit{maximal dependency cluster} of \(\Gtask\) if for all \(i,j \in \V^M\), \(i\) and \(j\) are connected\footnote{Here, \(i\) and \(j\) are considered connected if there is an undirected path along the directed edges in \(\Etask\) connecting nodes \(i\) and \(j\).}, and for all \(q \in \V \setminus \V^M\), \(i \in \V^M\), \(i\) and \(q\) are not connected. Thus, there are no task dependencies between different maximal dependency clusters.

We define \(\Gcom(\V, \Ecom)\) as the \textit{communication graph} of the agents, where \(\Ecom \subseteq \V \times \V\) is the set of undirected edges, consisting of unordered pairs \((i,j) \in \Ecom\) for all \(i,j \in \V\), indicating communication links between agents \(i\) and \(j\). Additionally, \(\Ncom_i = \{j \in \V \mid (i,j) \in \Ecom\}\) denotes the set of agent \(i\)’s neighbors in the undirected communication graph \(\Gcom\).  A communication link between agents \(i\) and \(j\) signifies they can exchange local information. For each maximal dependency cluster in \(\Gtask(\V, \Etask)\), we associate a communication subgraph \( \mathcal{G}^{c}_i(\V_i^M, \mathcal{E}_i^{c}) \), for \(i = 1, \ldots, m_{dc}\), where \(m_{dc} \leq N\) represents the number of clusters, and \(\V_i^M \subseteq \V\) and \(\Ecom_i \subseteq \Ecom\).

\begin{assumption} \label{assu:connceted_commu_graph}
	The communication subgraphs $\mathcal{G}^{c}_i(\V_i^M, \mathcal{E}_i^{c})$, $i = 1, \ldots, m_{dc}$, corresponding to maximal dependency clusters are undirected and connected.
\end{assumption}

\begin{example} \label{ex:graphs_example}
	Consider a MAS with \( N = 7 \) agents and the following long-term constraints: For agent 1 (\( m_1 = 2, \I_1=\{2\} \)): \( \psi_{1,1}(\bx_1) > 0 \), \( \psi_{1,2}(\bx_1, \bx_2) > 0 \); for agent 2 (\( m_2 = 2, \I_2=\{1\} \)): \( \psi_{2,1}(\bx_2) > 0 \), \( \psi_{2,2}(\bx_2, \bx_1) > 0 \); for agent 3 (\( m_3 = 1, \I_3=\{1\} \)): \( \psi_{3,1}(\bx_3, \bx_1) > 0 \); for agent 4 (\( m_4 = 2, \I_4=\{1,2,3\} \)): \( \psi_{4,1}(\bx_4, \bx_1) > 0 \), \( \psi_{4,2}(\bx_4, \bx_2, \bx_3) > 0 \); for agent 5 (\( m_5 = 3, \I_5=\{3,4\} \)): \( \psi_{5,1}(\bx_5) > 0 \), \( \psi_{5,2}(\bx_5, \bx_3) > 0 \), \( \psi_{5,3}(\bx_5, \bx_4) > 0 \); for agent 6 (\( m_6 = 2, \I_6=\{7\} \)): \( \psi_{6,1}(\bx_6) > 0 \), \( \psi_{6,2}(\bx_6, \bx_7) > 0 \); and for agent 7 (\( m_7 = 1, \I_7=\emptyset \)): \( \psi_{7,1}(\bx_7) > 0 \). Fig.\ref{fig:task_depen_graph} illustrates the directed task dependency graph \(\Gtask(\V, \Etask)\) of the MAS. Based on the inter-agent constraints, two maximal dependency clusters are identified: \(\V_1^M = \{1, 2, 3, 4, 5\}\) and \(\V_2^M = \{6, 7\}\). Fig. \ref{fig:commun_graph} illustrates a possible undirected communication graph for the system. In accordance with \Cref{assu:connceted_commu_graph}, any fixed, undirected, and connected communication graph is valid for each cluster. As depicted, agent 4 lacks a direct communication link with agents 1 and 2, despite needing to satisfy certain constraints related to them. The same holds for agents 5 and 3. Conversely, even though no tasks exist between agents 2 and 3, they can still be neighbors in the communication graph.
\end{example}

\begin{figure}[!tbp]
	\centering
	\begin{subfigure}[t]{0.55\linewidth}
		\centering
		\includegraphics[width=\linewidth]{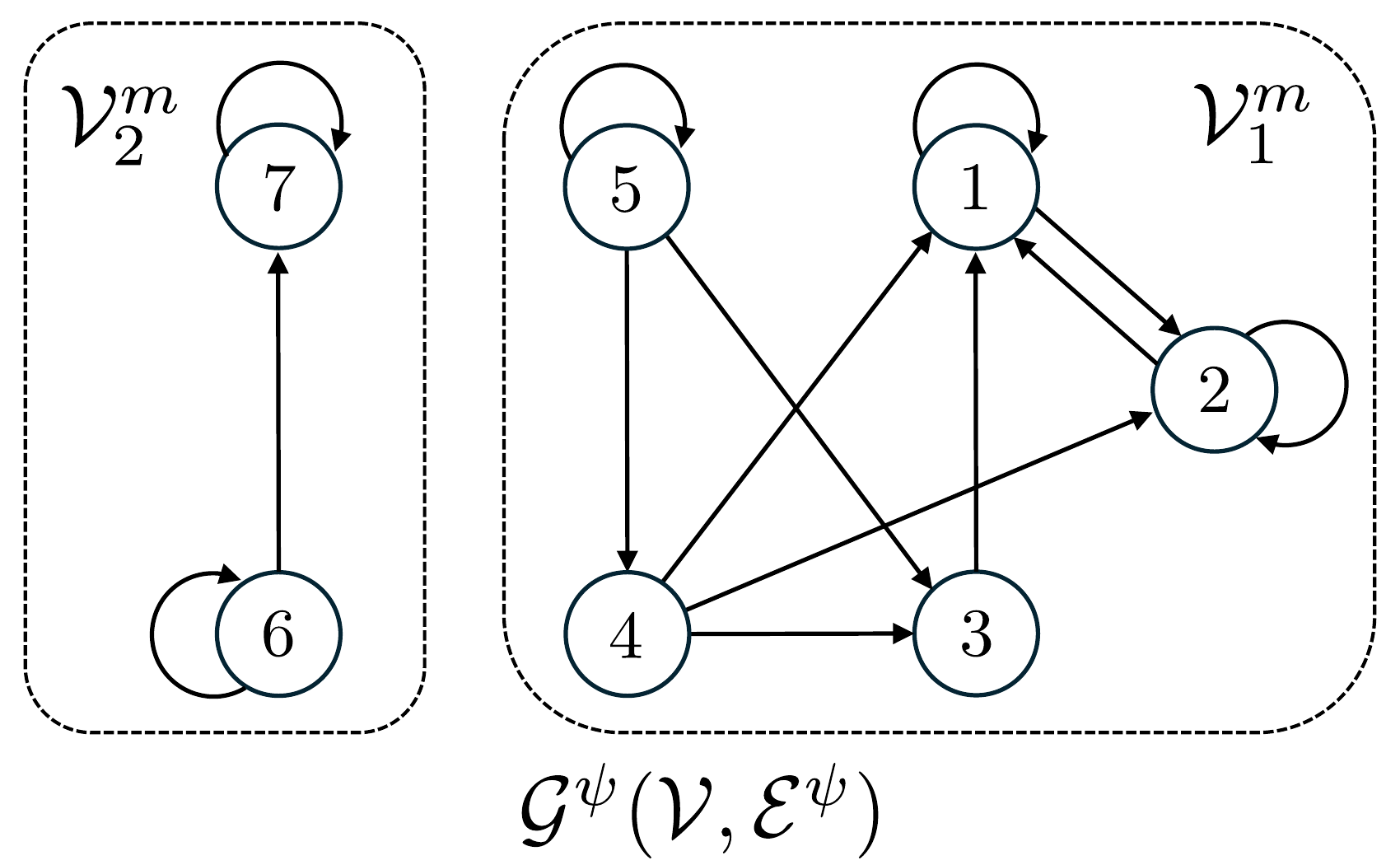}
		\caption{}
		\label{fig:task_depen_graph}
	\end{subfigure}%
	~~
	\begin{subfigure}[t]{0.435\linewidth}
		\centering
		\includegraphics[width=\linewidth]{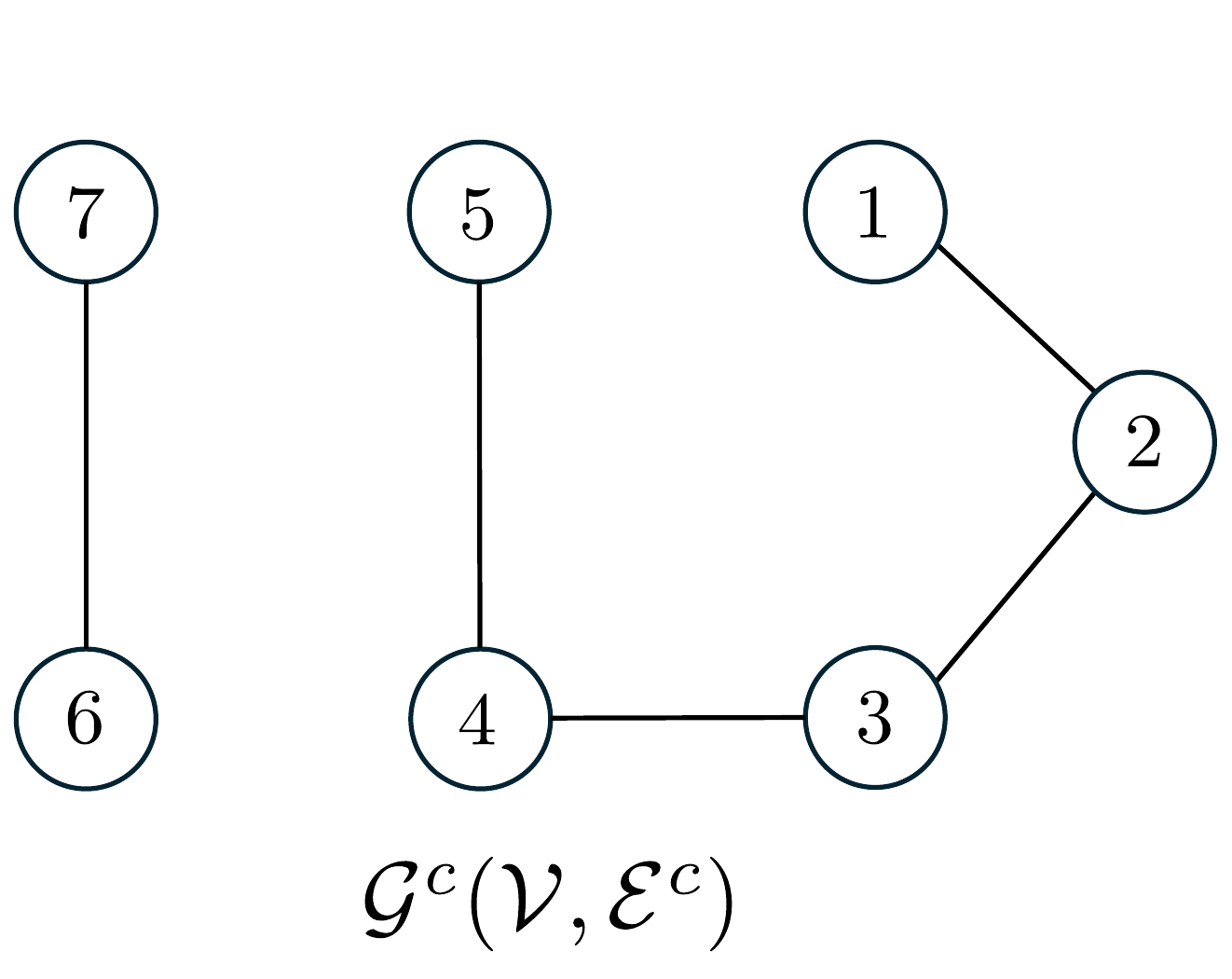}
		\caption{}
		\label{fig:commun_graph}
	\end{subfigure}
	\caption{(a) Directed task dependency graph composing of two maximal dependency clusters. (b) Undirected communication graph of multi-agent system with two connected components corresponding to each maximal dependency cluster in the task dependency graph. \vspace{-0.5cm}}
	\label{fig:task&commun_graphs}
\end{figure}

\begin{remark}
	As highlighted in Example \ref{ex:graphs_example}, Assumption \ref{assu:connceted_commu_graph} does not require neighboring agents in the task dependency graph $\Gtask$ to be neighbors in the communication graph $\Gcom$, or vice versa.
\end{remark}

Since the spatial constraints (tasks) within each maximal dependency cluster are independent of others, coordination of agents in each cluster can be handled separately. Hence, without loss of generality, we assume the following:
\begin{assumption}\label{assu:one_maximal_cluster}
	The task dependency graph $\Gtask(\V, \Etask)$ is composed of one maximal dependency cluster.
\end{assumption}

\begin{remark} \label{rem:shared_constraints}
	Neighboring agents \(i\) and \(j\) in the task dependency graph \(\mathcal{G}^d\), may or may not share a common constraint. For example, consider the coupled constraints \(\psi_{1,2}(\bx_1, \bx_2)\) and \(\psi_{2,2}(\bx_2, \bx_1)\) of agents 1 and 2  in \Cref{ex:graphs_example}. If \(\psi_{1,2}(\bx_1, \bx_2) \equiv \psi_{2,2}(\bx_2, \bx_1)\), i.e., both constraints are identical, the agents collaborate to satisfy it. If \(\psi_{1,2}(\bx_1, \bx_2) \not\equiv \psi_{2,2}(\bx_2, \bx_1)\), agent 1 is responsible for \(\psi_{1,2}(\bx_1, \bx_2) > 0\), while agent 2 is responsible for \(\psi_{2,2}(\bx_2, \bx_1) > 0\).
\end{remark}

%------------------------------------------------------------------------------------------
\subsection{Reformulating the Problem as a Centralized Optimization} 
\label{subsubsec:Centralized_Opt}

Inspired by the method in \cite{mehdifar2023control, mehdifar2024low}, we consolidate each agent's constraints into a single constraint as follows:
\begin{equation} \label{eq:agent_i_consolidated_const}
	\alphb_i(\bx_i,\bxc{i}) > 0, \quad i = 1, \ldots, N,
\end{equation}
where $\alphb_i: \R^d \times \R^{n_i d} \rightarrow \R$ is given by:
\begin{flalign} \label{eq:Consolidated_Const_Fun}
	&\alphb_i(\bx_i,\bxc{i}) \coloneqq \min \{ \psi_{i,1}(\bx_i,\bxc{i}), \ldots , \psi_{i,m_i}(\bx_i,\bxc{i}) \}.\!\!\!\!\!\!& 
\end{flalign}
Clearly, if \( \alphb_i(\bx_i,\bxc{i}) \leq 0 \), then at least one of agent \(i\)'s constraints is violated. We call \( \alphb_i(\bx_i,\bxc{i}) \) the \textit{consolidated constraint (task) function} of agent \(i\).

We can now merge the consolidated constraints of all $N$ agents into a single global constraint, representing satisfaction of the spatial constraints for the entire MAS as follows:
\begin{equation}\label{eq:MAS_consolidated_const}
	\betab(\bx) \coloneqq \min \{ \alphb_1(\bx_1,\bxc{1}), \ldots , \alphb_N(\bx_N,\bxc{N}) \} > 0.
\end{equation}
Note that when \(\betab(\bx) \leq 0\), at least one agent fails to meet one or more of its constraints. We refer to \(\betab(\bx)\) as the \textit{global constraint (task) function of the MAS}. Both functions in \eqref{eq:Consolidated_Const_Fun} and \eqref{eq:MAS_consolidated_const} are continuous, yet typically nonsmooth.

If $\betab(\mathdot)$ has compact level sets (see Lemma \ref{lem:radially_minus_beta}), then since it is continuous, \cite[Proposition 2.10]{Grippo2023intro} guarantees that $\betab(\mathdot)$ possesses at least one global maximizer. Thus, we can define:
\begin{equation}\label{eq:MAS_CollectiveConst_MAX}
	\betabstr \coloneqq \max_{\bx} \betab(\bx), 
\end{equation}
and
\begin{equation} \label{arg_centeral_opt}
	\bxs = \arg\max_{\bx} \betab(\bx),
\end{equation}
where $\bxs$ represents the optimal configuration of agents, maximizing the MAS's global constraint function. Here, \(\betab(\bx)\) serves as the objective function in \eqref{arg_centeral_opt}, with its optimal value indicating how well the spatial constraints are satisfied. If $\betabstr = \betab(\bxs) > 0$, all agents' constraints are \textit{collectively feasible}, and \(\bxs\) satisfies all constraints. If $\betabstr = \betab(\bxs) \leq 0$, at least one agent fails to meet one or more of its constraints, making \(\bxs\) the \textit{least-violating multi-agent formation} that maximizes $\betab(\mathdot)$. Thus, solving \eqref{arg_centeral_opt} promotes cooperative behavior among agents to effectively satisfy both individual and coupled long-term constraints in the MAS. Note that solving the optimization problem \eqref{arg_centeral_opt} requires information from all agents, leading to a centralized solution.

\begin{assumption} \label{assu:compact_level_set_beta_bar}
	At least one of the MAS constraint functions $\psi_{i,k}(\bx_i,\bxc{i}), k = 1, \ldots, m_i, i = 1, \ldots, N$, approaches $-\infty$ along any trajectory in $\R^{N d}$ on which $\|\bx\| \rightarrow +\infty$.
\end{assumption}

\begin{remark} \label{rem:aux_const}
	Assumption \ref{assu:compact_level_set_beta_bar} is not restrictive in practice and typically holds for well-posed multi-agent spatial constraints. It can always be satisfied by introducing an auxiliary individual constraint for each agent, \(\psi_{i, \text{aux}}(\bx_i) \coloneqq c_{\mathrm{aux}} - \|\bx_i\|^2 > 0, i = 1 , \ldots, N\), where \(c_{\mathrm{aux}} > 0\) is a sufficiently large constant. This constraint defines a large ball around the origin in the MAS operating space, encompassing all other spatial constraints of agent \(i\) without interfering with them.
\end{remark}

As the following lemma shows, Assumption \ref{assu:compact_level_set_beta_bar} is necessary and sufficient for the level sets of \(\betab(\bx)\) to be compact.

\begin{lemma} \label{lem:radially_minus_beta}
	Function $-\betab(\mathdot)$ (and thus $\betab(\mathdot)$) has compact level sets if and only if Assumption \ref{assu:compact_level_set_beta_bar} holds. 
\end{lemma}
\begin{proof}
	Since \(\betab(\cdot)\) is continuous, from \cite[Proposition 2.9]{Grippo2023intro}, \( -\betab(\bx) \) has compact level sets if and only if it is radially unbounded. From \eqref{eq:MAS_consolidated_const}, we can express \( -\betab(\bx) = \max \{-\alphb_1(\bx_1,\bxc{1}), \ldots, -\alphb_N(\bx_N,\bxc{N}) \} \). Thus, \( -\betab(\bx) \) is radially unbounded if and only if at least one of the functions \( -\alphb_i(\bx_i,\bxc{i}), i = 1, \ldots, N \), tends to \( +\infty \) along any trajectory where \( \|\bx\| \to +\infty \). Furthermore, from \eqref{eq:Consolidated_Const_Fun}, we have \( -\alphb_i(\bx_i,\bxc{i}) = \max \{-\psi_{i,1}(\bx_i,\bxc{i}), \ldots, -\psi_{i,m_i}(\bx_i,\bxc{i}) \} \). Therefore, \( -\alphb_i(\bx_i,\bxc{i}) \to +\infty \) if and only if at least one of \( -\psi_{i,k}(\bx_i,\bxc{i}), k = 1, \ldots, m_i \), tends to \( +\infty \). Consequently, \( -\betab(\bx) \) is radially unbounded if and only if at least one  \( \psi_{i,k}(\bx_i,\bxc{i}), k = 1, \ldots, m_i, i = 1, \ldots, N \), tends to \( -\infty \) along any trajectory where \( \|\bx\| \to +\infty \).	
\end{proof}

In the next subsection, we discuss how optimization problem \eqref{arg_centeral_opt} can be approximately solved in a distributed way, allowing agents to determine their optimal positions \(\bxs_i, i=1,\ldots,N\), using only local information and without sharing their constraint functions.

%------------------------------------------------------------------------------------------
\subsection{Multi-Agent Coordination via Distributed Optimization}
\label{Subsec:distOpt_scheme}

We propose a DO scheme to approximate the centralized nonsmooth problem in \eqref{arg_centeral_opt}. The key is leveraging the Log-Sum-Exp (LSE) function, which delivers a smooth under-approximation of the \(\min\) operators in \eqref{eq:Consolidated_Const_Fun} and \eqref{eq:MAS_consolidated_const}.

For each agent, we define:
\begin{subequations}\label{smooth_alph}
	\begin{align}
		\alpha_i(\bx_i,\bxc{i}) &\coloneqq -\frac{1}{\nualp} \ln \Big( \sum_{k=1}^{m_i} e^{- \nualp \, \psi_{i,k}(\bx_i,\bxc{i})} \Big) \label{smooth_alph_def} \\
		&\leq \alphb_i(\bx_i,\bxc{i}) \label{smooth_alph_underapp} \\
		&\leq \alpha_i(\bx_i,\bxc{i}) + \frac{1}{\nualp} \ln(m_i), \label{smooth_alph_ineq}
	\end{align}
\end{subequations}
where \(\alpha_i(\bx_i,\bxc{i})\) provides a smooth under-approximation of agent \(i\)'s consolidated constraint function \(\alphb_i(\bx_i,\bxc{i})\) in \eqref{eq:Consolidated_Const_Fun}, and \(\nualp > 0\) is a tuning parameter. As \(\nualp\) increases, the approximation improves, i.e., \(\alpha_i(\bx_i,\bxc{i}) \rightarrow \alphb_i(\bx_i,\bxc{i})\) as \(\nualp \rightarrow \infty\). From \eqref{smooth_alph}, it follows that \(\alpha_i(\bx_i,\bxc{i}) > 0\) is sufficient for \(\alphb_i(\bx_i,\bxc{i}) > 0\). Thus, when \(\alpha_i(\bx_i,\bxc{i}) > 0\), \(\bx_i\) and \(\bxc{i}\) ensure the satisfaction of agent \(i\)'s constraints.

We now define:
\begin{subequations}\label{eq:betabar_underapp}
	\begin{align}
		\betab_{\mathrm{ua}}(\bx) &\coloneqq \min \{ \alpha_1(\bx_1,\bxc{1}), \ldots , \alpha_N(\bx_N,\bxc{N}) \} \label{eq:Collective_Const_fun_underapp_nonsmooth} \\
		&\leq \betab(\bx) \\
		&\leq \min \{ \alpha_1(\bx_1,\bxc{1}) + \tfrac{1}{\nualp} \ln(m_1), \ldots \nonumber \\ 
		&\quad \quad \quad \; \; ,\alpha_N(\bx_N,\bxc{N}) + \tfrac{1}{\nualp} \ln(m_N) \} \\
		& \leq \betab_{\mathrm{ua}}(\bx) + \frac{1}{\nualp} \ln(\bar{m}),
	\end{align}
\end{subequations}
where $\bar{m} \coloneqq \max \{m_1, \ldots, m_N\}$. 
%Here, $\betab_{\mathrm{ua}}(\bx)$ is a nonsmooth under-approximation of $\betab(\bx)$, with improved accuracy as $\nualp$ increases. 
To introduce a smooth under-approximation of \eqref{eq:MAS_consolidated_const}, we replace the $\min$ operator in \eqref{eq:Collective_Const_fun_underapp_nonsmooth} with the LSE function, yielding:
\begin{subequations}\label{eq:beta_smooth_fun}
	\begin{align}
		\beta(\bx) &\coloneqq -\frac{1}{\nubet} \ln \Big( \sum_{i=1}^{N} e^{ - \nubet \alpha_i(\bx_i,\bxc{i}) } \Big) \label{eq:beta} \\
		&\leq \betab_{\mathrm{ua}}(\bx) \leq \beta(\bx) + \frac{1}{\nubet} \ln(N),
	\end{align}
\end{subequations}
where $\nubet > 0$ is a tuning parameter such that $\beta(\bx) \rightarrow \betab_{\mathrm{ua}}(\bx)$ as $\nubet \rightarrow \infty$. From \eqref{eq:betabar_underapp} and \eqref{eq:beta_smooth_fun}, we have:
\begin{equation} \label{eq:smooth under app_bounds}
	\beta(\bx) \leq \betab(\bx) \leq \beta(\bx) + \frac{1}{\nubet} \ln(N) + \frac{1}{\nualp} \ln(\bar{m}).
\end{equation}
Thus, $\beta(\bx)$ in \eqref{eq:beta} serves as a smooth under-approximation of the MAS's global constraint function $\betab(\bx)$ in \eqref{eq:MAS_consolidated_const}, and its accuracy can be improved by tuning $\nubet > 0$ and $\nualp > 0$. 

From \eqref{eq:smooth under app_bounds} and Lemma \ref{lem:radially_minus_beta} one can conclude that $-\beta(\mathdot)$ is radially unbounded and thus $\beta(\mathdot)$ has compact level sets. As a result, in lieu of \cite[Proposition 2.10]{Grippo2023intro}, we define:
\begin{subequations} \label{eq:beta_max_bounds}
	\begin{align}
		\betastr &\coloneqq \max_{\bx} \beta(\bx) \label{eq:beta_max} \\
		&\leq \betabstr \leq \betastr + \frac{1}{\nubet} \ln(N) + \frac{1}{\nualp} \ln(\bar{m}).
	\end{align}
\end{subequations}
From \eqref{eq:beta_max}, it is evident that $\betastr > 0$ ensures the feasibility of MAS constraints. Conversely, $\betastr \leq - \frac{1}{\nubet} \ln(N) - \frac{1}{\nualp} \ln(\bar{m})$ is sufficient to indicate the infeasibility of MAS constraints.

From \eqref{eq:beta_max_bounds}, we can verify that the centralized optimization problem in \eqref{eq:MAS_CollectiveConst_MAX} can be approximately solved by maximizing $\beta(\bx)$. Consequently, the optimal agent positions (formation) for collaboratively satisfying the spatial multi-agent constraints can be approximated by solving:
\begin{equation} \label{arg_opt_smooth}
	\bxsh \coloneqq \arg\max \beta(\bx),
\end{equation}
where $\bxsh$ is the approximated optimal multi-agent formation such that $\|\bxs - \bxsh\| \leq \epsilon$. From \eqref{eq:beta_max_bounds}, it is guaranteed that $\epsilon \rightarrow 0$ as $\nualp, \nubet \rightarrow +\infty$, since $\betastr \rightarrow \betabstr$. Although smaller values of $\nualp$ and $\nubet$ increase the gap between $\betastr$ and $\betabstr$, this does not necessarily imply that $\epsilon \geq 0$ will grow. In fact, $\epsilon$ may remain small even for low values of $\nualp$ and $\nubet$. In other words, the distance between the optimal points $\bxs$ and $\bxsh$ may remain small even when $\betabstr$ and $\betastr$ differ substantially.

Apart from ensuring the smoothness of the optimization problems in \eqref{eq:beta_max} and \eqref{arg_opt_smooth}, $\beta(\bx)$ also facilitates solving \eqref{arg_opt_smooth} in a distributed setting. To see this, observe that since $\ln(\cdot)$ is strictly increasing and $\sum_{i=1}^{N} e^{- \nubet \alpha_i(\bx_i,\bxc{i})}$ is strictly positive, we have:
\begin{align}\label{eq:argmax_distir_opt}
	\bxsh &= \arg\max -\frac{1}{\nubet} \ln \Big( \sum_{i=1}^{N} e^{- \nubet \alpha_i(\bx_i,\bxc{i})} \Big) \nonumber \\
	&= \arg\min \sum_{i=1}^{N} e^{- \nubet \alpha_i(\bx_i,\bxc{i})}.
\end{align}

Now define 
\begin{equation} \label{eq:h_fun}
	h_i(\bx_i, \bxc{i}) \coloneqq  \sum_{k = 1}^{m_i} e^{-\nualp \psi_{i,k}(\bx_i, \bxc{i})}, \quad i = 1, \ldots, N. 
\end{equation}	
Substituting \eqref{smooth_alph_def} into \eqref{eq:argmax_distir_opt} and using \eqref{eq:h_fun} yields: 
\begin{equation} \label{eq:intermed_cost}
		\bxsh = \arg\min \sum_{i=1}^{N} h_i(\bx_i, \bxc{i})^{\frac{\nubet}{\nualp}}.
\end{equation}
Consequently, defining
\begin{equation}  \label{eq:local_cost}
	f_i(\bx_i,\bxc{i}) \coloneqq h_i(\bx_i, \bxc{i})^{\frac{\nubet}{\nualp}}, \quad i = 1, \ldots, N,
\end{equation}
leads to
\begin{equation}\label{eq:f_fun}
	f(\bx) \coloneqq \sum_{i = 1}^{N} f_i(\bx_i,\bxc{i}),
\end{equation} 
representing the global objective function of the MAS, which is composed of the sum of agents' local objective functions $f_i(\bx_i,\bxc{i})$. As a result, $\bxsh$ is the minimizer of both \eqref{arg_opt_smooth} and the following optimization problem:
\begin{equation}\label{eq:distri_opt_problem}
	\min_{\bx} f(\bx) = \sum_{i = 1}^{N} f_i(\bx_i,\bxc{i}),
\end{equation}
which aligns with DO setups \cite{gharesifard2013distributed,yang2019survey}.

Note that each agent's objective function $f_i(\bx_i,\bxc{i})$ depends solely on its own constraint functions $\psi_{i,k}(\bx_i,\bxc{i})$, for $k = 1, \ldots, m_i$. As a result, based on the constraints in \eqref{eq:agent_const_predicateform}, a private local objective function, defined in \eqref{eq:local_cost}, can be formulated for each agent. By solving the smooth unconstrained DO problem \eqref{eq:distri_opt_problem}, agents can collaboratively reach a formation that meets both individual and coupled constraints in MAS as much as possible. In particular, the feasibility of the MAS constraints under the multi-agent formation \(\bxsh\) obtained from the above optimization problem is determined by checking whether \(\betab(\bxsh) > 0\) (feasible) or \(\betab(\bxsh) \leq 0\) (infeasible).

\begin{remark} \label{rem:noncooperative_game}
	Each agent's consolidated constraint function $\alphb_i(\bx_i,\bxc{i})$ in \eqref{eq:Consolidated_Const_Fun}, or its smooth under-approximation $\alpha_i(\bx_i,\bxc{i})$ in \eqref{smooth_alph_def}, can be interpreted as a local benefit function. By maximizing it, each agent seeks to satisfy its individual and coupled constraints as much as possible in an egoistic manner. Since each agent's benefit function depends on the positions of other agents, this egoistic optimization can be analyzed through noncooperative game theory \cite{hespanha2017noncooperative}, potentially leading to Nash equilibria. However, this approach does not always ensure the satisfaction of feasible multi-agent spatial constraints. The DO problem in \eqref{eq:distri_opt_problem} resolves this by collaboratively maximizing the smallest $\alphb_i(\bx_i,\bxc{i})$ across the MAS, as shown in \eqref{eq:MAS_consolidated_const} and \eqref{eq:MAS_CollectiveConst_MAX}.
\end{remark}

\begin{remark}
	We note that our DO problem in \eqref{eq:distri_opt_problem} closely resembles a Continuous Distributed Constraint Optimization Problem (C-DCOP)\footnote{In DCOP literature, a “constraint” denotes cost functions that determine how agents select decision variable values and should not be confused with constraints in the optimization literature.} in MAS \cite{fransman2023distributed,hoang2020new,fioretto2018distributed}. In C-DCOPs, it is typically assumed that each agent controls a decision variable and that the edges of an undirected constraint graph model shared objective functions among neighboring agents, with the goal of minimizing the sum of these functions over all edges. Unlike C-DCOPs, our approach does not require neighboring agents to share a common objective function. Instead, each agent $i$ has its own private local objective function that may depend on the decision variables of some other agents (i.e., neighbors in the directed task dependency graph \(\Gtask\)), yet $f_i$ remains unknown to them.
\end{remark}

%==========================================================================================
\subsection{Properties of the Multi-Agent Objective Function \eqref{eq:f_fun}}
\label{Subsect:costfun_prop}

To identify suitable algorithms for solving the DO problem \eqref{eq:distri_opt_problem}, it is crucial to examine the properties of the global multi-agent objective function $f(\bx)$. The following lemma provides a sufficient condition for the log-convexity of $f(\bx)$ in \eqref{eq:f_fun}, which implies its convexity.
\begin{lemma} \label{lem:objective_fun_convexity}
	If all agents' constraint functions in \eqref{eq:agent_const_predicateform}, i.e., $\psi_{i,k}(\bx_i, \bxc{i})$, $k = 1, \ldots$, $m_i, i = 1, \ldots, N$ are concave, then the global objective function $f(\bx)$ in \eqref{eq:f_fun} is log-convex.
\end{lemma}
\begin{proof}
	Consider \( f(\bx) \coloneqq \sum_{i = 1}^{N} f_i(\bx_i, \bxc{i}) \) in \eqref{eq:f_fun}. The convexity of all local objective functions \( f_i(\bx_i, \bxc{i}) \) with respect to their arguments ensures the convexity of \( f(\bx) \). The functions \( e^{-\nualp \psi_{i,k}(\bx_i, \bxc{i})}, k = 1,\ldots,m_i, i = 1,\ldots,N \), are log-convex \cite[Section 3.5]{boyd2004convex} in their arguments when $\psi_{i,k}(\bx_i, \bxc{i})$ are concave w.r.t. their arguments. From Lemma \ref{lem:log_cvx_perserve} we know that the sum of log-convex functions remains log-convex, thus \( h_i(\bx_i,\bxc{i}) = \sum_{k = 1}^{m_i} e^{-\nualp \psi_{i,k}(\bx_i,\bxc{i})}, i = 1,\ldots,N \), are log-convex in their arguments. Furthermore, $f_i(\bx_i, \bxc{i}) = h_i(\bx_i,\bxc{i})^{\frac{\nubet}{\nualp}}$ with $\frac{\nubet}{\nualp} > 0$ is also log-convex since $\log (h_i(\bx_i,\bxc{i})^{\frac{\nubet}{\nualp}} ) = \frac{\nubet}{\nualp} \log \left(h_i(\bx_i,\bxc{i}) \right)$ is convex due to the log-convexity of $h_i(\bx_i,\bxc{i})$. Consequently, $f_i(\bx_i, \bxc{i}), i = 1,\ldots,N$, are log-convex in their arguments. Finally, as sum of log-convex functions is log-convex then $f(\bx)$ in \eqref{eq:f_fun} is log-convex.
\end{proof}

\begin{remark}\label{rem:psi constraints concavity}
	The concavity of the multi-agent constraint functions $\psi_{i,k}(\bx_i,\bxc{i})$ limits the types of constraints that can be considered for agents. For example, since any $p$-norm is convex, one can verify that continuously differentiable inter-agent connectivity constraints, i.e., $R^2 - \|\bx_i - \bx_j\|^2 > 0$, where $R > 0$ represents the maximum distance between agents $i$ and $j$, are allowed. However, collision avoidance constraints, i.e., $\|\bx_i - \bx_j\|^2 - r^2 > 0$, where $r > 0$ denotes the minimum distance between agents $i$ and $j$, do not meet this requirement and may result in a nonconvex $f(\bx)$ in \eqref{eq:f_fun}.
\end{remark}

The convexity of \( f(\bx) \) in \eqref{eq:f_fun} is crucial for many DO algorithms, but identifying conditions for strict convexity, which guarantees a unique minimizer, is equally important. This is addressed in the following lemma.

\begin{lemma} \label{lem:strict_cvx_global_obj_fun}
Let all $\psi_{i,k}(\bx_i, \bxc{i})$ in \eqref{eq:agent_const_predicateform} be concave functions. If at least one of the constraint functions \( \psi_{i,k}(\bx_i,\bxc{i}) \), $k = 1, \ldots, m_i$, for each agent is strictly concave, then the global objective function \( f(\bx) \) in \eqref{eq:f_fun} is strictly log-convex.
\end{lemma}
\begin{proof}
	Recall from the proof of Lemma \ref{lem:objective_fun_convexity} that all local objective functions \( f_i(\bx_i, \bxc{i}) \) are log-convex under the concavity of $\psi_{i,k}(\bx_i, \bxc{i})$, \( k = 1, \ldots, m_i, i = 1, \ldots, N \). 

	Note that the local functions \( f_i(\bx_i,\bxc{i}) \), \( i = 1, \ldots, N \), typically depend only on a subset of the components of \( \bx \). Thus, even if one function is strictly (log-) convex, this does not guarantee the strict (log-) convexity of \( f(\bx) = \sum_{i = 1}^{N} f_i(\bx_i, \bxc{i}) \). However, if each \( f_i(\bx_i, \bxc{i}) \) is strictly (log-) convex in \( \bx_i \), then \( f(\bx)\) is ensured to be strictly (log-) convex.
	
	All constraint functions of agent \( i \), \( \psi_{i,k}(\bx_i,\bxc{i}) \), \( k = 1, \ldots, m_i \), inherently depend on \( \bx_i \), even when \( \I_i = \emptyset \), meaning agent \( i \) only has individual constraints. If at least one function \( \psi_{i,k}(\bx_i, \bxc{i}) \) is strictly concave in (all of) its arguments, it is also strictly concave with respect to \( \bx_i \). Thus, we focus on establishing the strict (log-) convexity of each \( f_i(\bx_i, \bxc{i}) \) in \( \bx_i \), which is sufficient to ensure the strict (log-) convexity of the global objective function \( f(\bx) \). Hereafter, we omit \( \bxc{i} \) from the function arguments for notational simplicity when no ambiguity arises.
	
	Consider functions \( h_i(\bx_i, \mathdot) \), \( i = 1, \ldots, N \), as defined in \eqref{eq:h_fun}. For any \( 0 \leq \theta \leq 1 \), we have:
	\begin{equation}\label{eq:h_i_cvx_start}
		h_i(\theta \bx_i + (1-\theta) \by_i, \mathdot)  =  \sum_{k=1}^{m_i} e^{-\nualp \psi_{i,k}(\theta \bx_i + (1-\theta) \by_i, \mathdot)},
	\end{equation} 
	for all $i = 1, \ldots, N$. From the concavity of agent \( i \)'s constraint functions \( \psi_{i,k}(\bx_i,\bxc{i}) \), we know that for all \( k \in \{1, \ldots, m_i\} \), the following inequality holds:
	\begin{equation*}\label{eq:psi_i_concave}
		\psi_{i,k}(\theta \bx_i + (1-\theta) \by_i, \mathdot) \geq \theta \psi_{i,k}(\bx_i , \mathdot) + (1-\theta) \psi_{i,k}(\by_i, \mathdot),
	\end{equation*}
	which indicates the concavity of $\psi_{i,k}(\bx_i , \mathdot)$ functions in their first argument. Furthermore, by assumption agent \( i \) has at least one strictly concave constraint function. This guarantees the existence of at least one \( k^\prime \in \{1, \ldots, m_i\} \) for which the following strict inequality holds:
	\begin{equation*}\label{eq:psi_i__strictly_concave}
		\psi_{i,k^\prime}(\theta \bx_i + (1-\theta) \by_i, \mathdot) > \theta \psi_{i,k^\prime}(\bx_i , \mathdot) + (1-\theta) \psi_{i,k^\prime}(\by_i,\mathdot).
	\end{equation*}
	Since $e^{-\nualp \psi_{i,k}(\theta \bx_i + (1-\theta) \by_i, \mathdot)}, k = 1, \ldots, m_i$, are strictly decreasing, from the above inequalities and \eqref{eq:h_i_cvx_start} we get:
	\begin{align}\label{eq:h_i_cvx_second}
		&\resizebox{.99\hsize}{!}{$ h_i(\theta \bx_i + (1-\theta) \by_i, \mathdot) < \sum_{k=1}^{m_i} e^{-\nualp \theta \psi_{i,k}(\bx_i , \mathdot)} \, e^{-\nualp  (1-\theta) \psi_{i,k}(\by_i, \mathdot)} $} \nonumber \\
		&\leq \Big( \sum_{k=1}^{m_i} e^{-\nualp \psi_{i,k}(\bx_i , \mathdot)} \Big)^\theta \, \Big( \sum_{k=1}^{m_i} e^{-\nualp \psi_{i,k}(\by_i, \mathdot)} \Big)^{1-\theta} \nonumber \\
		&= h_i(\bx_i, \mathdot)^\theta \, h_i(\by_i, \mathdot)^{1-\theta}, \quad i = 1, \ldots, N, 
	\end{align}
	where, Lemma \ref{lem:holder_ineq} (Hölder's inequality) with \( 1/p = \theta \) is applied to obtain the final inequality. As evident from \eqref{eq:h_i_cvx_second}, the functions \( h_i(\bx_i, \mathdot) \) for \( i = 1, \ldots, N \), are strictly log-convex in their first arguments. Additionally, since \( h_i(\bx_i, \mathdot) \), \( i = 1, \ldots, N \), are strictly positive, \eqref{eq:h_i_cvx_second} and \eqref{eq:local_cost} lead to:
	\begin{align}\label{eq:f_i_strict_cvx}
		f_i(\theta \bx_i + (1-\theta) \by_i, \mathdot) &= \left( h_i(\theta \bx_i + (1-\theta) \by_i, \mathdot) \right) ^{\frac{\nubet}{\nualp}} \nonumber \\
		&< h_i(\bx_i, \mathdot)^{\frac{\nubet}{\nualp}\theta} \, h_i(\by_i, \mathdot)^{\frac{\nubet}{\nualp} (1-\theta)} \nonumber \\
		&= f_i(\bx_i, \mathdot)^{\theta} \, f_i(\by_i, \mathdot)^{1-\theta},  
	\end{align}
	for all $i = 1, \ldots, N$, establishing the strict log-convexity of all \( f_i(\bx_i, \mathdot) \) with respect to their first argument \cite{boyd2004convex, constantin2018convex}. As a result, we can conclude that \( f(\bx) \coloneqq \sum_{i = 1}^{N} f_i(\bx_i, \bxc{i}) \) is also strictly log-convex and thus strictly convex.
\end{proof}

\begin{remark} \label{rem:strict_cvx_not_restrictive}
	We emphasize that ensuring the strict (log-) convexity of the global objective function \( f(\bx) \) is not restrictive in practice. When all multi-agent constraint functions are concave but do not meet the additional condition in Lemma \ref{lem:strict_cvx_global_obj_fun}, an auxiliary individual constraint can be added for each agent to guarantee the strict (log-) convexity of \( f(\bx) \). Specifically, one can introduce strictly concave constraints \( \psi_{i, \text{aux}}(\bx_i) \coloneqq c_{\mathrm{aux}} - \|\bx_i\|^2 > 0, \, i = 1, \ldots, N \), with a sufficiently large \( c_{\mathrm{aux}} > 0 \). As stated in Remark \ref{rem:aux_const}, these constraints do not interfere with the original MAS constraints.
\end{remark}

The closed-form gradient of agent \( i \)'s local objective function \( f_i(\bx_i, \bxc{i}) \) defined in \eqref{eq:local_cost} with respect to \( \bx \) is derived as follows (function arguments omitted for simplicity):
\begin{align} \label{eq:gard_f_i}
	\nabla f_i = \frac{\partial f_i}{\partial \bx}^\top &= \frac{\nubet}{\nualp} \; h_i^{(\frac{\nubet}{\nualp}-1)} \; \frac{\partial h_i}{\partial \bx}^\top  \nonumber \\
	&= - \nubet \, h_i^{(\frac{\nubet}{\nualp}-1)} \, \sum_{k = 1}^{m_i} \frac{\partial \psi_{i,k}}{\partial \bx}^\top e^{-\nualp \psi_{i,k}}, 
\end{align} 
for all $i = 1 , \ldots, N$. Finally, it is not difficult to conclude that $f(\bx)$ and $\nabla f(\bx) = \sum_{i = 1}^{N} \nabla f_i$ are locally Lipschitz.

%==========================================================================================
\subsection{Distributed Optimization Algorithm and Control Design}
\label{Subsect:dist_opt_alg_ctrl}

First, note that, with a slight abuse of notation, one can write \(f_i(\bx_i, \bxc{i})\) in \eqref{eq:distri_opt_problem} as \(f_i(\bx)\). Let \(\xest{j}{i}\) denote agent \(i\)'s local estimate of \(\bx_j\). Moreover, let $\xestvec{i} \coloneqq \col(\xest{j}{i})_{j = 1, \ldots, N} \in \R^{Nd},$ represent agent \(i\)'s local estimate of \(\bx\). Without loss of generality, we assume that \(\xest{i}{i}(0) = \bx_i(0)\) for all \(i = 1, \ldots, N\), that is, each agent initializes its position estimate with its true initial position.

To solve the unconstrained optimization problem \eqref{eq:distri_opt_problem} in a distributed manner, one can reformulate it as an equivalent consensus-based optimization problem, as shown in \cite[Lemma 3.1]{gharesifard2013distributed}, as follows:
\begin{flalign} \label{eq:dist_opt_reformu_standard}
	&\text{minimize} \; \ftil(\bxtil) = \sum_{i = 1}^{N} f_i(\xestvec{i}), \; \text{subject to} \;  \bL \bxtil = \mathbf{0}_{N^2 d}, \!\!\!\!\!\!\!&
\end{flalign}
where $\bxtil = \col(\xestvec{i})_{i=1,\ldots, N} \in \R^{N^2 d}$ represents the stacked vector of all agents' estimates, and $\bL \coloneqq \Lapla \otimes \mbI_{N d} \in \R^{N^2d \times N^2d}$, where $\Lapla \in \R^{N \times N}$ is the Laplacian matrix \cite{mesbahi2010graph} corresponding to the undirected, connected inter-agent communication graph $\Gcom$. Notice that $\bL \bxtil = \mathbf{0}_{N^2 d}$ holds when $\xestvec{i} = \xestvec{j}$ for all $i, j \in \V$. In other words, the optimal solution of \eqref{eq:dist_opt_reformu_standard} must satisfy $\bxtils = \mathbf{1}_N \otimes \bxsh$ for all $\bxsh \in \R^{N d}$.

For a strictly convex function $\ftil(\bxtil)$, the solution to the consensus-based DO problem \eqref{eq:dist_opt_reformu_standard} can be derived using the continuous-time distributed algorithm proposed in \cite{kia2015distributed}:
\begin{subequations} \label{eq:dist_opt_protocol_compact}
	\begin{align}
		\bxtildot &= - k_1 \, \nabla_{\bxtil} \ftil(\bxtil) - k_2 \, \bL \bxtil - \bz, \\
		\dot{\bz} &= k_1\, k_2 \, \bL \bxtil,
	\end{align}
\end{subequations}
with $\sum_{i=1}^{N} \bz_i(0) = \mathbf{0}_{N d}$, where $\bz \coloneqq [\bz_1^\top, \ldots, \bz_N^\top]^\top \in \R^{N^2 d}$ and $k_1, k_2 > 0$ are some positive gains. Note that \eqref{eq:dist_opt_protocol_compact} leads to the following update law for each agent: 
\begin{subequations} \label{eq:dist_opt_protocol}
	\begin{align}
		\xestvecdot{i} &= - k_1 \nabla f_i(\xestvec{i}) - k_2  \sum_{j \in \Ncom_i}^{} (\xestvec{i} - \xestvec{j}) - \bz_i, \\
		\dot{\bz}_i &= k_1 k_2 \sum_{j \in \Ncom_i}^{} (\xestvec{i} - \xestvec{j}). 
	\end{align}
\end{subequations}

To implement the above distributed protocol, each agent must compute the gradient of its local objective function and share its solution estimate \(\xestvec{i}\) with neighbors in the communication graph \(\Gcom(\V, \Ecom)\). By \eqref{eq:dist_opt_protocol} and the initial condition \(\xest{i}{i}(0) = \bx_i(0)\), each agent's controller is given by:
\begin{equation} \label{eq:kinematic_control}
	\bu_i = (\mbe_N^{i} \otimes \mbI_d)^\top  \xestvecdot{i}.
\end{equation}

\begin{theorem} \label{th:main_cvx_only}
	Consider the multi-agent system in \eqref{eq:agent_single_integ_dyn} with continuously differentiable long-term spatial constraints in \eqref{eq:agent_const_predicateform}. Under Assumptions 1-3 and Lemma \ref{lem:strict_cvx_global_obj_fun}, for any $\xestvec{i}(0)$ with $\xestvec{i}_i(0) = \bx_i(0)$ and $\bz_i(0) = \mathbf{0}_{N^2 d}$, for all  $i = 1, \ldots, N$, the distributed control protocol \eqref{eq:dist_opt_protocol} and \eqref{eq:kinematic_control} ensures that all agents asymptotically converge to their optimal positions, i.e., $\bx(t) \rightarrow \bxsh$ as $t \rightarrow \infty$.
\end{theorem}
\begin{proof}
	The multi-agent communication graph $\Gcom$ is assumed to be undirected and connected, as per Assumptions \ref{assu:connceted_commu_graph} and \ref{assu:one_maximal_cluster}. From Assumption \ref{assu:compact_level_set_beta_bar}, the compactness of level sets for $f(\bx)$ in \eqref{eq:f_fun}, and thus for $\ftil(\bxtil)$ in \eqref{eq:dist_opt_reformu_standard}, is inferred. Additionally, Lemma \ref{lem:strict_cvx_global_obj_fun} guarantees the strict convexity of $f(\bx)$, and hence the strict convexity of $\ftil(\bxtil)$. As a result, asymptotic convergence $\xestvec{i}(t) \rightarrow \bxsh$ for all $i = 1, \ldots, N$ is ensured, following \cite[Theorem 8]{kia2015distributed}. Consequently, owing to $\xestvec{i}_i(0) = \bx_i(0)$,  applying \eqref{eq:kinematic_control} to each agent ensures $\bx(t) \rightarrow \bxsh$ as $t \rightarrow \infty$.		
\end{proof}

\begin{remark}
		Ensuring strict convexity of $f(\bx)$ is straightforward (see Lemma \ref{lem:strict_cvx_global_obj_fun} and Remark \ref{rem:strict_cvx_not_restrictive}). However, for a convex $f(\bx)$, one can use the proportional-integral (also known as primal-dual \cite{wang2011control,jakovetic2020primal}) continuous-time DO algorithm from \cite{gharesifard2013distributed}, instead of \eqref{eq:dist_opt_protocol}. This method, unlike \eqref{eq:dist_opt_protocol}, requires agents to exchange both solution estimate vectors and dual variable vectors through the communication network, making it more resource-intensive. More precisely, the augmented Lagrangian associated with the equality-constrained optimization problem \eqref{eq:dist_opt_reformu_standard} can be written as $F(\bxtil,\bz) = \ftil(\bxtil) + \bxtil^\top L \, \bz + \frac{1}{2} \bxtil^\top L \, \bxtil$, where $\bz \coloneqq [\bz_1^\top, \ldots, \bz_N^\top]^\top \in \R^{N^2 d}$ stacks all dual variables, and $\bz_i \in \R^{Nd}$ denotes the dual vector maintained by agent $i$. The solution to the consensus-based distributed optimization problem \eqref{eq:dist_opt_reformu_standard} can then be obtained via the continuous-time saddle-point flow of \( F(\bxtil,\bz) \) \cite{gharesifard2013distributed}:
	\begin{subequations} \label{eq:saddle_point_flow_dyn}
		\begin{align*}
			\dot{\bxtil} &= - \nabla_{\bxtil}\, F(\bxtil,\bz) = -\nabla \ftil(\bxtil) - L \bz - L \bxtil, \\
			\dot{\bz} &= \nabla_{\bz}\, F(\bxtil,\bz) = L \bxtil.
		\end{align*}
	\end{subequations}
	This yields the following proportional–integral (PI) continuous-time distributed optimization protocol for each agent:
	\begin{subequations} \label{eq:dist_opt_protocol}
		\begin{align*}
			\xestvecdot{i} &= - \nabla f_i(\xestvec{i}) - \sum_{j \in \Ncom_i}^{} (\bz_i -\bz_j) - \sum_{j \in \Ncom_i}^{} (\xestvec{i} - \xestvec{j}), \\
			\dot{\bz}_i &= \sum_{j \in \Ncom_i}^{} (\xestvec{i} - \xestvec{j}). 
		\end{align*}
	\end{subequations}
	Observe that this protocol requires each agent to communicate both its local solution estimate \(\xestvec{i}\) and its local dual variables $\bz_i \in \R^{Nd}$ to its neighbors. 
\end{remark}

\begin{remark}
	Note that although Theorem \ref{th:main_cvx_only} guarantees only the asymptotic convergence of agents' positions \(\bx(t)\) to their optimal positions \(\bxsh\), the multi-agent constraints may be satisfied at a finite, though unknown, time once \(\betab(\bx(t)) > 0\). Equivalently, it suffices to have $\beta(\bx(t)) > 0$, since this implies \(\betab(\bx(t)) > 0\). Hence, in practice, when the MAS constraints are feasible, exact convergence of $\bx(t)$ to $\bxsh$ is not required, and one may terminate the algorithm in \eqref{eq:dist_opt_protocol} in finite time before full convergence.
\end{remark}

%==========================================================================================
\subsection{Avoiding Potential Numerical Issues in Implementing \eqref{eq:dist_opt_protocol}}
\label{Subsect:numerical}

From \eqref{eq:h_fun} and \eqref{eq:local_cost}, it is evident that each agent's local objective function is a sum of exponential terms involving $\psi_{i,k}$ functions. Specifically, in any open subset of $\R^{(n_i+1)d}$ where agent $i$'s constraints are violated (i.e., $\psi_{i,k}(\bx_i, \bxc{i}) < 0$), the local objective function $f_i(\bx_i,\bxc{i})$ grows exponentially towards infinity. This behavior also applies to $\nabla f_i$, as seen in \eqref{eq:gard_f_i}. The increase in both $f_i(\bx_i,\bxc{i})$ and $\nabla f_i$ becomes more pronounced when $\frac{\nubet}{\nualp} > 1$. Consequently, the global objective function $f(\bx)$ in \eqref{eq:f_fun} also grows exponentially in regions where the multi-agent constraints are violated, i.e., when $\betab(\bx) < 0$ in \eqref{eq:MAS_consolidated_const}. Thus, while $f(\bx)$ and $f_i(\bx_i,\bxc{i})$, $i = 1, \ldots, N$, are well-defined over the entire spaces $\R^{N d}$ and $\R^{(n_i+1) d}$, as $\betab(\bx)$ decreases, $f(\bx)$ behaves similarly to a barrier function outside the bounded set \(\Omega \coloneqq \{\bx \mid \betab(\bx) > 0 \} \subset \R^{Nd}\). Therefore, if the initial values of $\xestvec{i}(0)$, $i = 1, \ldots, N$, in \eqref{eq:dist_opt_reformu_standard} are such that $\beta(\xestvec{i}(0)) \ll 0$, the proposed DO algorithm \eqref{eq:dist_opt_protocol} may face numerical overflow at $t = 0$. To mitigate this, we propose treating the tunable parameter $\nubet \in (0, +\infty)$ as a time-varying function. Specifically, based on the structure of \eqref{eq:local_cost}, initializing $\nubet(0)$ with a sufficiently small value (e.g., $\nubet(0) = 0.01$) and gradually increasing it to some sufficiently large finite value $\nubet^{\mathrm{final}} > \nubet(0)$ (e.g., $\nubet^{\mathrm{final}} > \ln(N)$) during the execution of \eqref{eq:dist_opt_reformu_standard} helps in preventing a potential numerical overflow at $t = 0$. Recall that even a small $\nubet^{\mathrm{final}}$ can still yield a good approximate solution $\bxsh$ to $\bxs$, as explained below \eqref{arg_opt_smooth}.

%%%%%%%%%%%%%%%%%%%%%%%%%%%%%%%%%%%%%%%%%%%%%%%%%%%%%%%%%%%%%%%%%%%%%%%%%%%%%%%%
\section{Simulation Results}
\label{sec:simu_results}

Consider \eqref{eq:agent_single_integ_dyn} with $N = 3$ and $d = 2$ under the DO-based control scheme \eqref{eq:dist_opt_protocol} and \eqref{eq:kinematic_control}. In all subsequent simulations, we initialize $\bz_i(0) = \mathbf{0}_{N d}$ for all $i = 1, \ldots, N$, with the initial positions $\bx_i(0)$ and the values of $\xestvec{i}(0)$ randomly chosen. The random initialization of $\xestvec{i}(0)$ is constrained by $\xest{i}{i}(0) = \bx_i(0)$. Moreover, we set $k_1 = k_2 = 1$ in \eqref{eq:dist_opt_protocol}. As discussed in Subsection \ref{Subsect:numerical}, $\nubet$ is treated as a slowly increasing parameter for implementing \eqref{eq:dist_opt_protocol}, where
\[
\nubet(t) = \begin{cases}
	0.022 t + 0.01 & 0 \leq t \leq T, \\
	\nubet^{\mathrm{nom}} & t > T
\end{cases},
\]
with $\nubet^{\mathrm{nom}} = 5$ and $T \coloneqq (\nubet^{\mathrm{nom}} - 0.01)/0.022$. Additionally, in \eqref{eq:h_fun} and \eqref{eq:local_cost}, we set $\nualp = 5$. The agents' communication graph is a line graph with the undirected edge set $\Ecom = \{(1,2), (2,3)\}$, for which $\Lapla = [1, -1, 0; -1, 2, -1; 0, -1, 1]$.

\vspace{0.2cm}
\noindent \textit{\textbf{Case A:}} \textit{Feasible Spatial Constraints (Consensus)}

In the first simulation example, the agents' long-term constraints, as specified in \eqref{eq:agent_const_predicateform}, are defined as follows. For agent 1: \( \psi_{1,1}(\bx_1) = 1 - \|\bx_1 - [2,0]^\top\|^2 > 0 \). For agent 2: \( \psi_{2,1}(\bx_2, \bx_1) = 1 - \|\bx_2 - \bx_1\|^2 > 0 \). For agent 3: \( \psi_{3,1}(\bx_3, \bx_2) = 1 - \|\bx_3 - \bx_2\|^2  > 0\). These long-term constraints require agent 1 to end up within a ball of radius 1 centered at \([2,0]\), agent 2 to maintain a maximum distance of 1 from agent 1, and agent 3 to remain within 1 distance unit from agent 2. These conditions collectively guide the agents to an optimal formation, which is expected to occur when all converge at \([2,0]\). 

Fig.\ref{fig:simu_exam1_consensus} (top-left) illustrates the agents' final positions after 300 seconds, confirming that consensus at \([2,0]\) is the optimal configuration for satisfying the spatial multi-agent constraints. Fig.\ref{fig:simu_exam1_consensus} (top-right) shows the evolution of $\betab(\bx)$ in \eqref{eq:MAS_consolidated_const} and $\beta(\bx)$ in \eqref{eq:beta}, evaluated at agents' positions. Recall that the solution to \eqref{eq:dist_opt_reformu_standard} maximizes $\beta(\bx)$, which also approximately maximizes $\betab(\bx)$. Thus, the evolution of $\beta(\bx)$ or $\betab(\bx)$ can verify convergence to the optimal solution of \eqref{eq:dist_opt_reformu_standard}. However, these values are shown for verification only, as they are unavailable to the agents and unused in the control scheme in \eqref{eq:dist_opt_protocol} and \eqref{eq:kinematic_control}. The convergence of $\betab(\bx)$ to a positive constant ($\betab(\bxsh)\approx 1$) indicates the feasibility of the multi-agent constraints. Additionally, all constraints are satisfied from $t = 16.15$ onward, as $\betab(\bx)$ remains positive for all $t \geq 16.15$. Fig.\ref{fig:simu_exam1_consensus} (bottom-left) depicts the evolution of the entries of the vector $\bxtil \in \R^{18}$ in \eqref{eq:dist_opt_reformu_standard}, partitioned as agents estimation vectors $\xestvec{i} \in \R^{6}, i = 1, 2, 3$. Fig. \ref{fig:simu_exam1_consensus} (bottom-right) presents $\bz_i$ evolution for each agent. Each vector entry is shown in a different color, with corresponding entries across agents using the same color but a different line style. The results confirm that all $\xestvec{i}, i = 1, 2, 3$, have reached consensus on the optimal solution $\bxsh$ of \eqref{eq:dist_opt_reformu_standard}.  

\vspace{0.2cm}
\noindent \textit{\textbf{Case B:}} \textit{More General Feasible Spatial Constraints}

In this case, the long-term constraints of the agents are considered as follows. For agent 1: \( \psi_{1,1}(\bx_1) = 1 - \|\bx_1 - [2,0]^\top\|^2  > 0 \). For agent 2: \( \psi_{2,1}(\bx_2, \bx_1) = 3^2 - \|\bx_2 - \bx_1\|^2 > 0 \), and \( \psi_{2,2}(\bx_2, \bx_3) = 2^2 - \|\bx_2 - \bx_3\|^2 > 0 \). For agent 3: \( \psi_{3,1}(\bx_3) = 1 - \|\bx_3 - [-2,0]^\top\|^2  > 0\).

Fig.\ref{fig:simu_exam2_feasible} summarizes the simulation results. The evolution of $\betab(\bx)$ confirms that the multi-agent constraints are collectively feasible, with all agents' constraints being satisfied from approximately \(t = 13.48\). In this example, after 300 seconds the agents reach the approximate optimal configuration \(\bxsh = [(\bxsh_1)^\top, (\bxsh_2)^\top, (\bxsh_3)^\top ]^\top \in \R^6\), where \(\bxsh_1 = [1.99, 0.00]\), \(\bxsh_2 = [-0.61, 0.03]\), and \(\bxsh_3 = [-1.99, 0.00]\). Clearly, agents 1 and 3 tend to converge to their target locations at \([2,0]\) and \([-2,0]\), best satisfy their individual constraints, while agent 2 tends to converge to a point that best satisfies its coupled constraints with agents 1 and 3.

\vspace{0.2cm}
\noindent \textit{\textbf{Case C:}} \textit{Tightly Feasible Spatial Constraints}

In this example, the long-term constraints of agents 1 and 3 are kept the same as in \textit{Case B} while agent 2's constraints are modified as follows: \( \psi_{2,1}(\bx_2, \bx_1) = (1.7)^2 - \|\bx_2 - \bx_1\|^2 > 0 \), and \( \psi_{2,2}(\bx_2, \bx_3) = (0.7)^2 - \|\bx_2 - \bx_3\|^2 > 0 \), enforcing a tighter feasible space for agent 2 compared with \textit{Case B}. 

Fig. \ref{fig:simu_exam3_tightfeasible} presents the simulation results for this case. Although the optimal value of $\beta(\bx)$ is negative ($\beta(\bxsh) \approx -0.1$), its maximizer $\bxsh$ yields a positive value for $\betab(\bx)$ ($\betab(\bxsh) \approx 0.1$), confirming that the multi-agent constraints are collectively feasible. Moreover, all agents' constraints are satisfied for \(t \geq 62.21\). It is also worth noting that the steady-state value of $\betab(\bx)$ in this example is much smaller than the value found in \textit{Case B} (which was approximately 1), as shown in Fig. \ref{fig:simu_exam2_feasible} (top-right). This indicates that in this case, the multi-agent spatial constraints are tightly feasible, and thus closer to an infeasible scenario.

In this simulation, after 300 seconds, the agents reached the approximate optimal configuration: \(\bxsh_1 = [1.10, 0.00]\), \(\bxsh_2 = [-0.49, 0.00]\), and \(\bxsh_3 = [-1.10, 0.00]\). Notably, agents 1 and 3 approach the boundary of their constraints, deviating from their target positions to assist agent 2 in satisfying its coupled constraints. Recall that agents 1 and 3 do not have access to agent 2's local objective function nor its constraints, and this collaborative behavior emerges from solving the optimization problem \eqref{eq:distri_opt_problem} in a distributed manner, solely by sharing local estimates \(\xestvec{i}, i = 1, 2, 3\), via the communication network. Hence, this approach is useful in practice, as it allows agents to collaborate without knowing each other's spatial constraints and preserves privacy.

\vspace{0.2cm}
\noindent \textit{\textbf{Case D:}} \textit{Infeasible Spatial Constraints}

In this example, agents 1 and 3 constraints are the same as \textit{Case B} and agent 2's constraints are further shrunk to \( \psi_{2,1}(\bx_2, \bx_1) = (1.4)^2 - \|\bx_2 - \bx_1\|^2 > 0 \), and \( \psi_{2,2}(\bx_2, \bx_3) = (0.4)^2 - \|\bx_2 - \bx_3\|^2 > 0 \), leading to an infeasible scenario. The simulation results are shown in Fig. \ref{fig:simu_exam4_infeasible}. As depicted in Fig. \ref{fig:simu_exam4_infeasible} (right), the maximum value of $\beta(\bx)$ is negative ($\beta(\bxsh) \approx -0.35$), and its optimizer $\bxsh$ gives $\betab(\bxsh) \approx -0.18$, indicating that the multi-agent constraints are infeasible. Despite this, the agents achieve a least-violating spatial configuration. After 300 seconds, the agents reached the approximate optimal configuration: $\bxsh_1 = [0.98, 0.00]$, $\bxsh_2 = [-0.98, 0.00]$, and $\bxsh_3 = [-0.41, 0.00]$, with none of the agents fully satisfying their constraints. Although agents 1 and 3 slightly violated their constraints, their adjustments contributed to minimizing overall violations of the constraints at the group level.

\vspace{0.2cm}
\noindent \textit{\textbf{Case E:}} \textit{Nonconvex Global Objective Function (Rendezvous)}

We consider the same long-term constraints as in \textit{Case A}, with the addition of the following constraints for agents 2 and 3: \( \psi_{2,2}(\bx_2, \bx_1) = \|\bx_2 - \bx_1\|^2 - (0.2)^2 > 0 \), \( \psi_{3,2}(\bx_3, \bx_1) = 1 - \|\bx_3 - \bx_1\|^2 > 0 \), \( \psi_{3,3}(\bx_3, \bx_1) = \|\bx_3 - \bx_1\|^2 - (0.2)^2 > 0 \), and \( \psi_{3,4}(\bx_3, \bx_2) = \|\bx_3 - \bx_2\|^2 - (0.2)^2 > 0 \). Note that \(\psi_{2,2} > 0\), \(\psi_{3,3} > 0\), and \(\psi_{3,4} > 0\), ensure a minimum separation distance among agents. As stated in Remark \ref{rem:psi constraints concavity}, these constraints can lead to a nonconvex global multi-agent objective function in \eqref{eq:f_fun}. As shown in Fig. \ref{fig:simu_exam5_rendezvous}, the agents achieve a sub-optimal formation (due to the nonconvexity of the objective function) that satisfies all spatial constraints roughly from \(t = 15.56\) onwards. Unlike \textit{Case A}, where position consensus occurs, the minimum distance constraints prevent consensus, but the agents still achieve rendezvous near the target point \([2,0]\), with agent 1 ending up exactly there. The final positions after 300 seconds are \(\bxsh_1 = [2.00, 0.00]\), \(\bxsh_2 = [1.33, 0.28]\), and \(\bxsh_3 = [1.42, -0.43]\). Note that the resulting formation is not unique.
\begin{figure}[!tbp]
	\centering
	\begin{subfigure}[t]{0.43\linewidth}
		\centering
		\includegraphics[width=\linewidth]{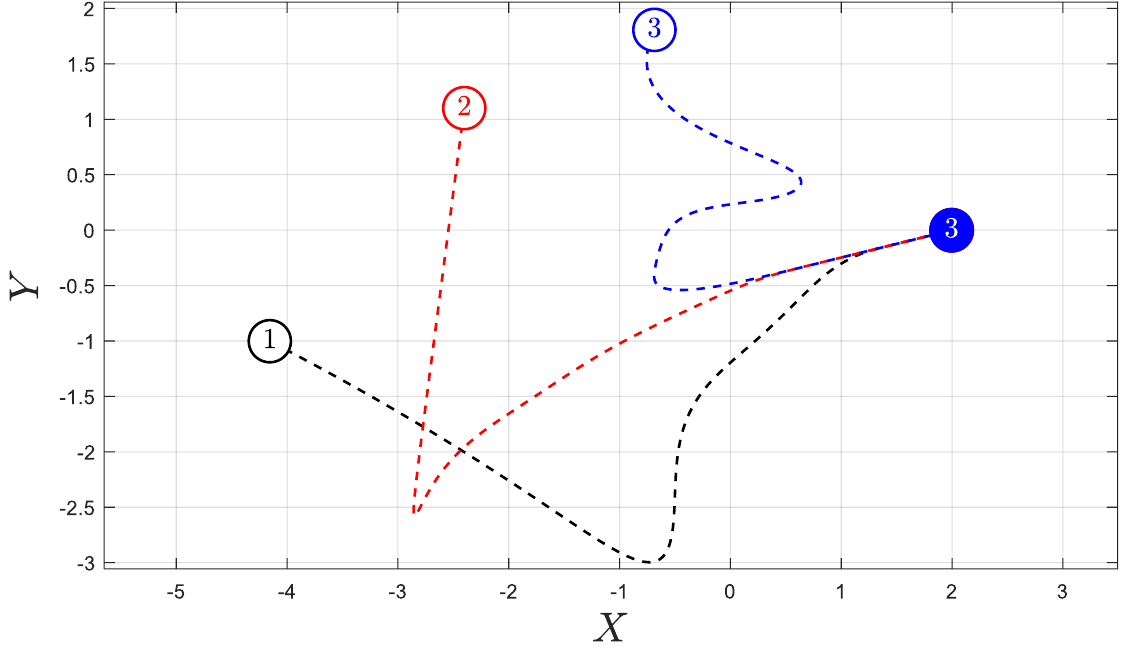}
	\end{subfigure}
	~~~~
	\begin{subfigure}[t]{0.48\linewidth}
		\centering
		\includegraphics[width=\linewidth]{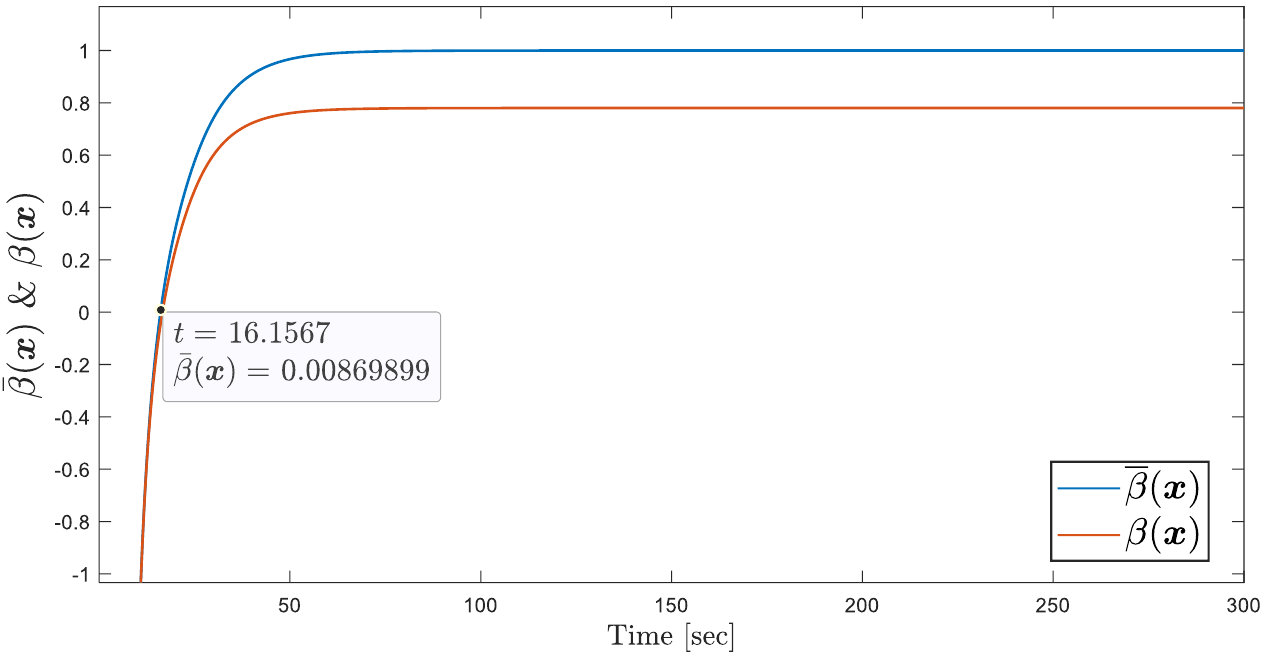}
	\end{subfigure}%
	
	\begin{subfigure}[t]{0.48\linewidth}
		\centering
		\includegraphics[width=\linewidth]{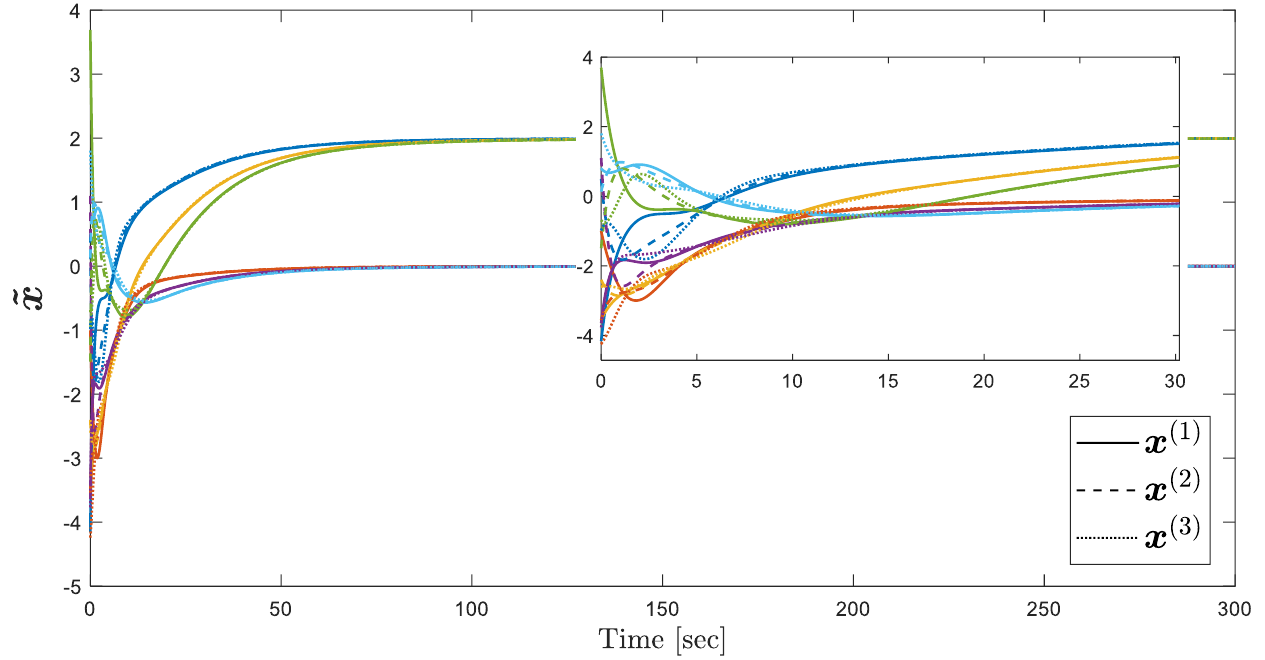}
	\end{subfigure}
	~
	\begin{subfigure}[t]{0.48\linewidth}
		\centering
		\includegraphics[width=\linewidth]{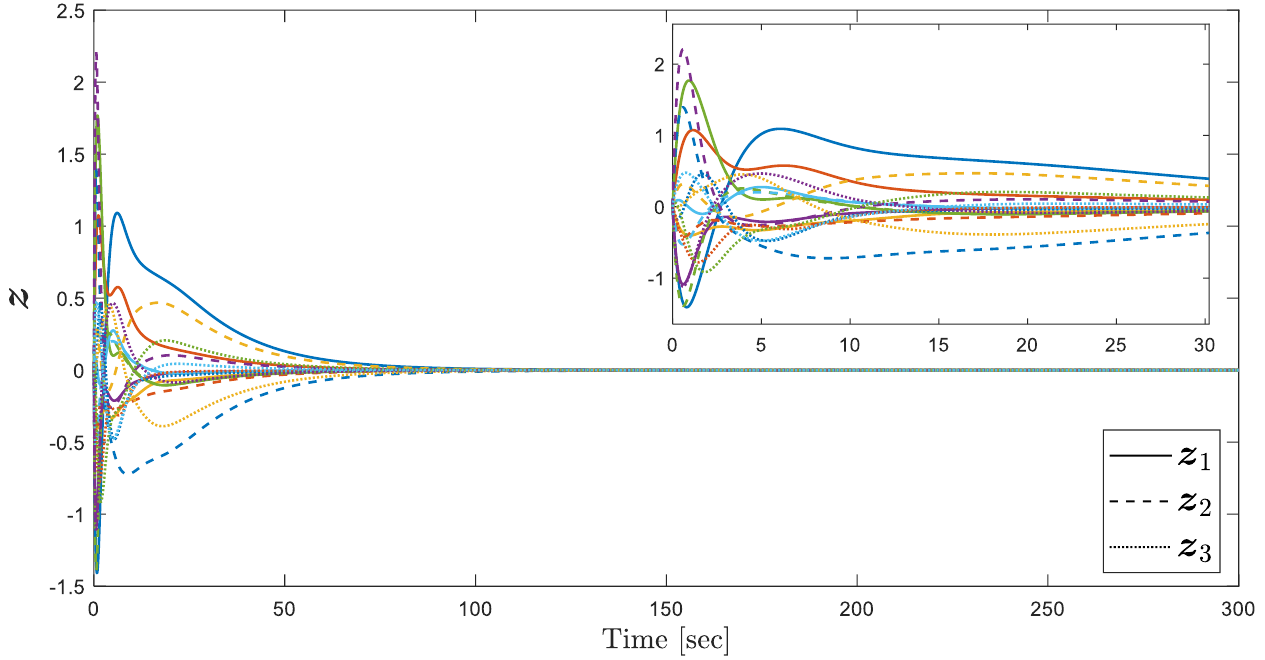}
	\end{subfigure}
	\caption{Simulation results under multi-agent spatial constraints of Case A. \vspace{-1.8cm}}
	\label{fig:simu_exam1_consensus}
\end{figure}

\begin{figure}[!tbp]
	\centering
	\begin{subfigure}[t]{0.43\linewidth}
		\centering
		\includegraphics[width=\linewidth]{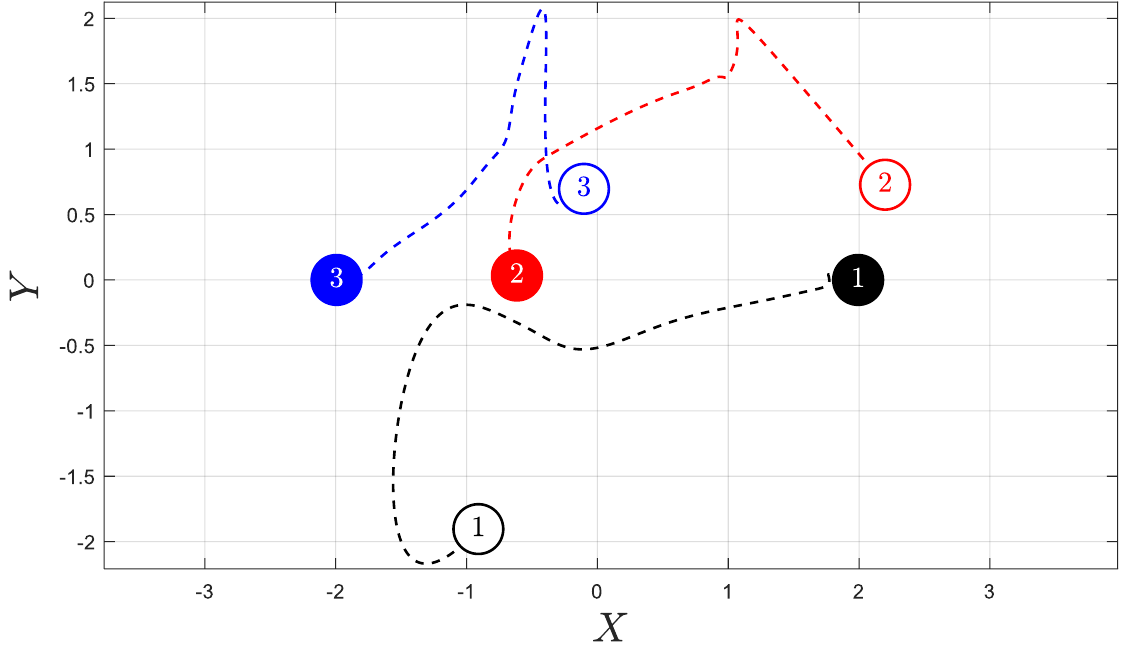}
	\end{subfigure}
	~~~~
	\begin{subfigure}[t]{0.48\linewidth}
		\centering
		\includegraphics[width=\linewidth]{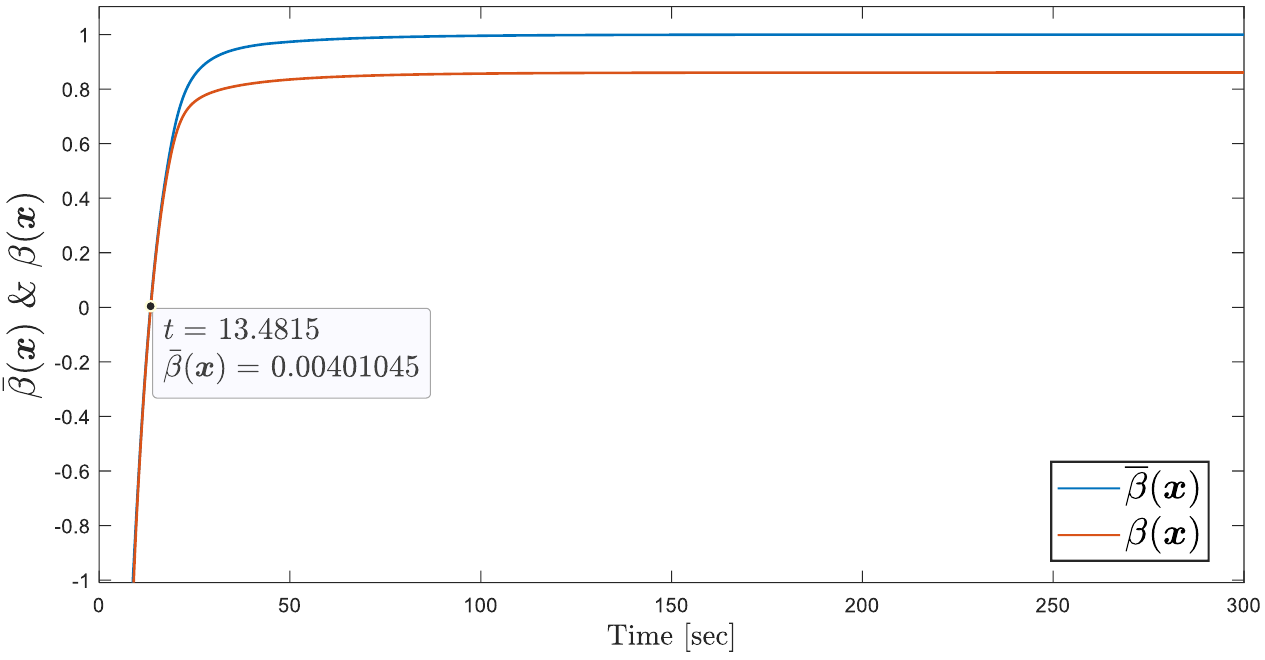}
	\end{subfigure}
	
	\begin{subfigure}[t]{0.48\linewidth}
		\centering
		\includegraphics[width=\linewidth]{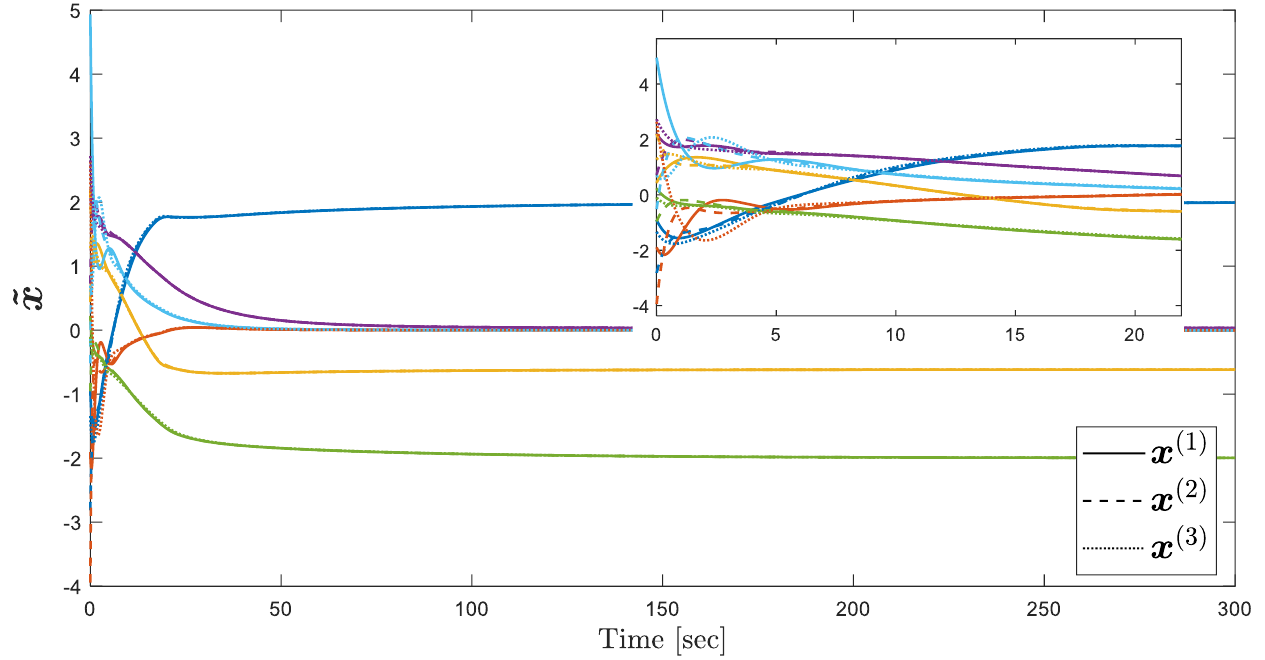}
	\end{subfigure}
	~
	\begin{subfigure}[t]{0.48\linewidth}
		\centering
		\includegraphics[width=\linewidth]{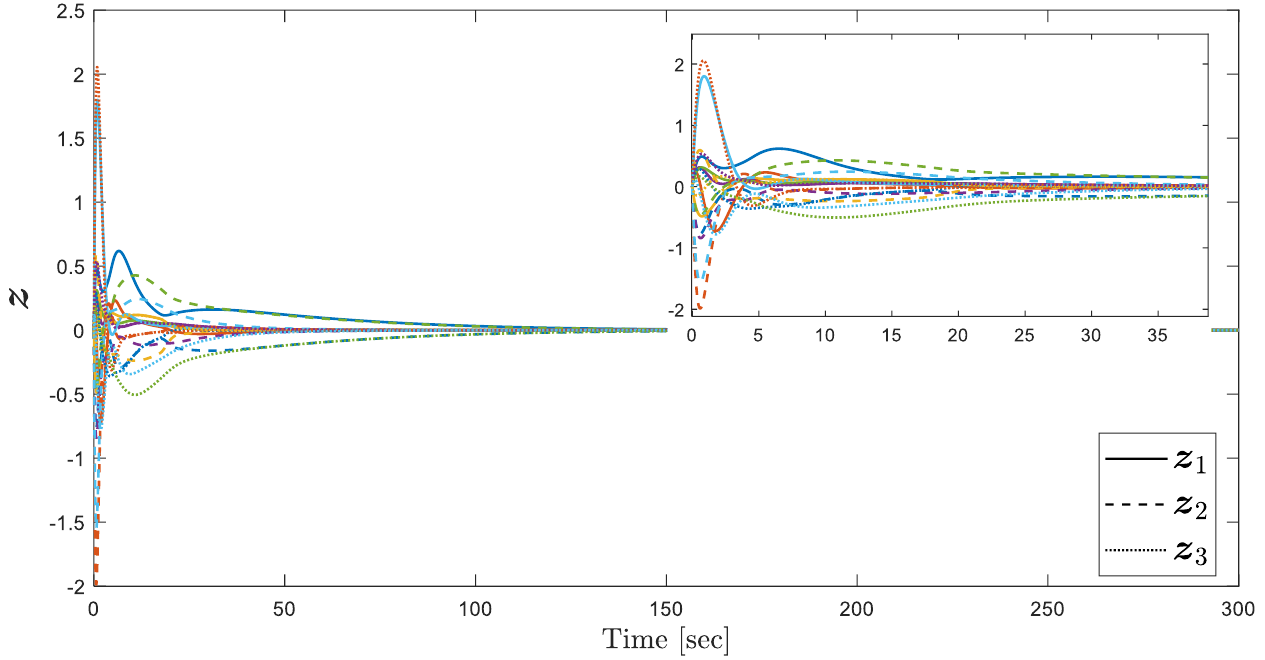}
	\end{subfigure}
	\caption{Simulation results under multi-agent spatial constraints of Case B. \vspace{-0.2cm}}
	\label{fig:simu_exam2_feasible}
\end{figure}

\begin{figure}[!tbp]
	\centering
	\begin{subfigure}[t]{0.43\linewidth}
		\centering
		\includegraphics[width=\linewidth]{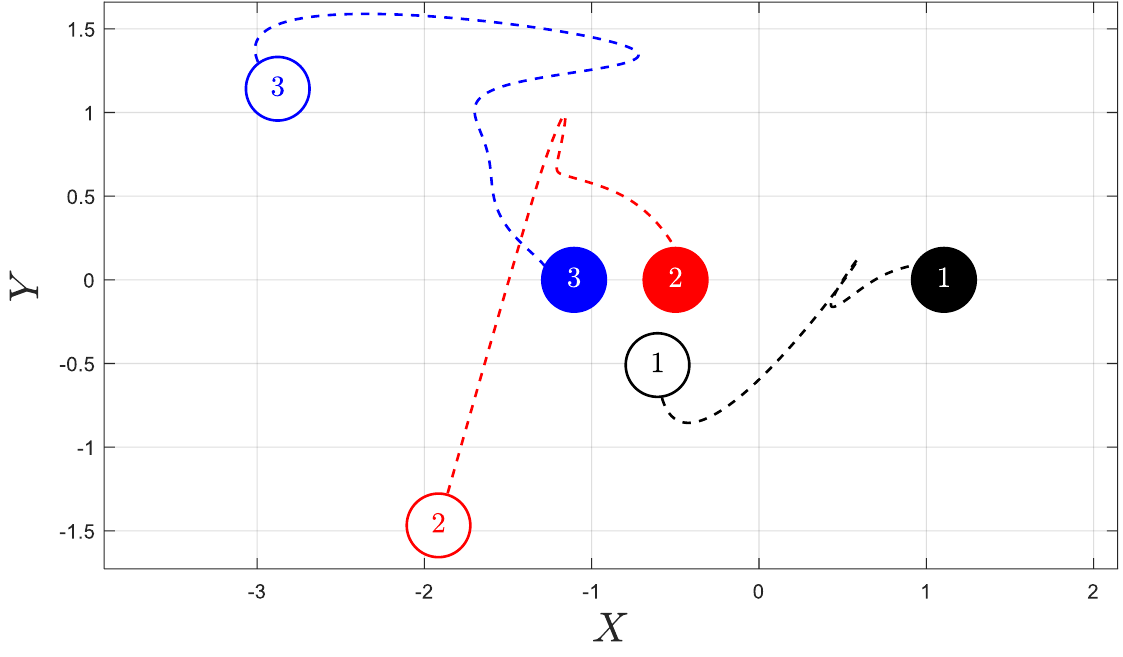}
	\end{subfigure}%
	~~~~~
	\begin{subfigure}[t]{0.48\linewidth}
		\centering
		\includegraphics[width=\linewidth]{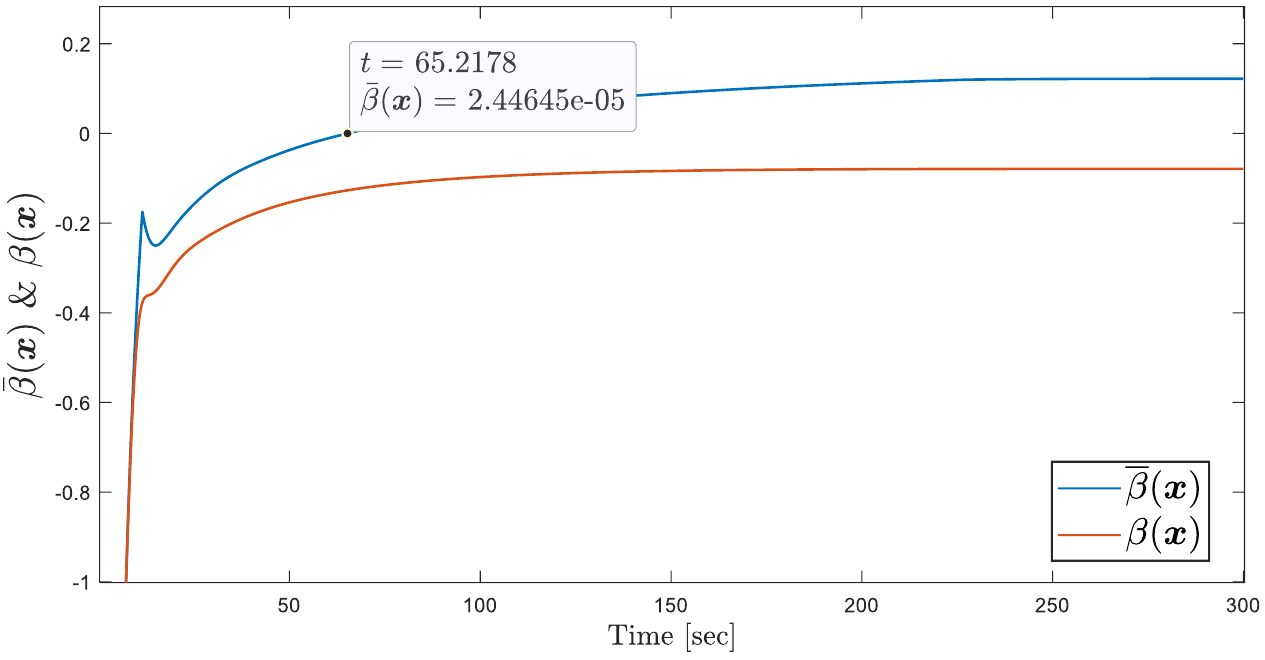}
	\end{subfigure}
	\caption{Simulation results under multi-agent spatial constraints of Case C.\vspace{-0.2cm}}
	\label{fig:simu_exam3_tightfeasible}
\end{figure}

\begin{figure}[!tbp]
	\centering
	\begin{subfigure}[t]{0.43\linewidth}
		\centering
		\includegraphics[width=\linewidth]{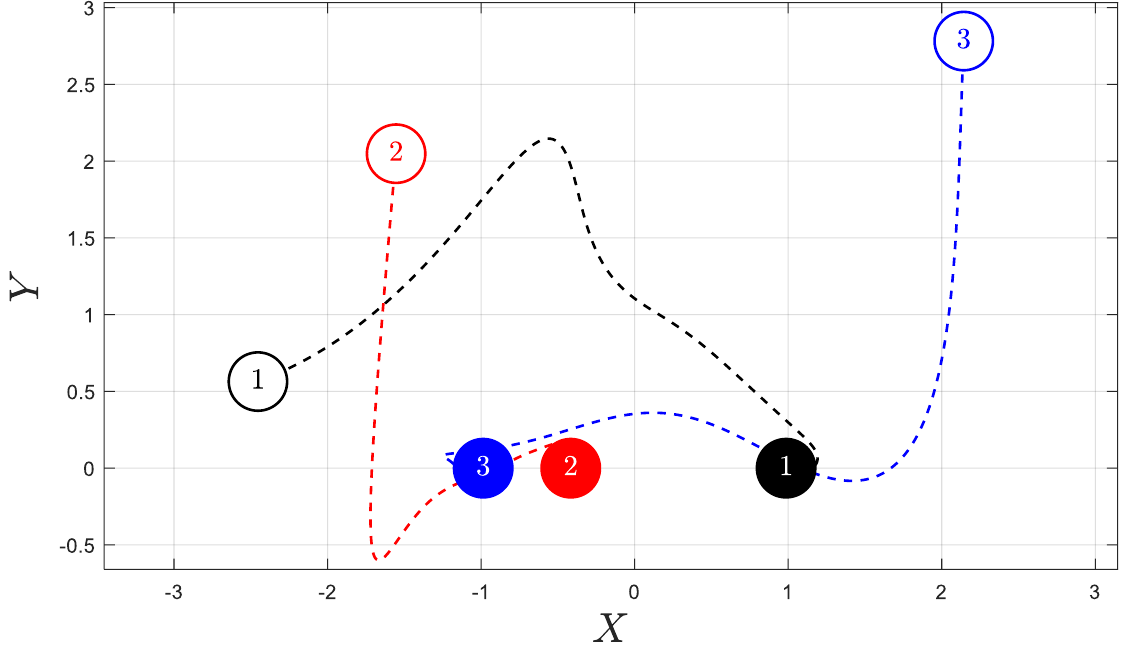}
	\end{subfigure}
	~~~~~
	\begin{subfigure}[t]{0.48\linewidth}
		\centering
		\includegraphics[width=\linewidth]{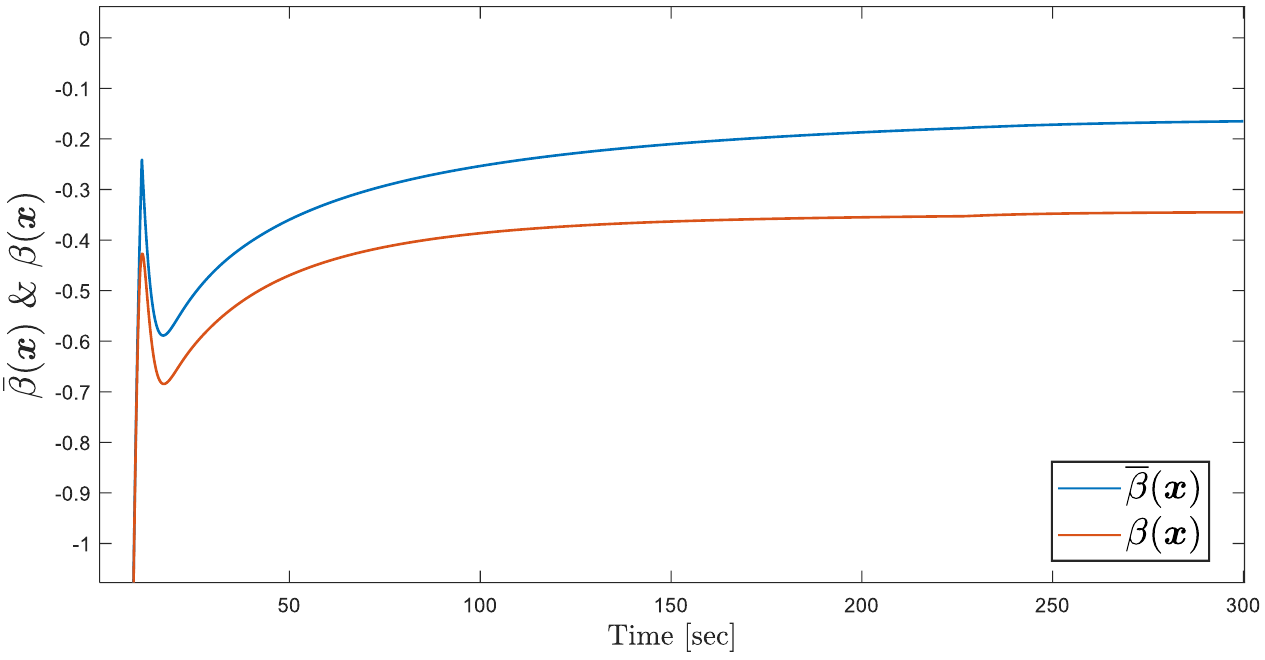}
	\end{subfigure}
	\caption{Simulation results under multi-agent spatial constraints of Case D.\vspace{-0.2cm}}
	\label{fig:simu_exam4_infeasible}
\end{figure}

\begin{figure}[!tbp]
	\centering
	\begin{subfigure}[t]{0.43\linewidth}
		\centering
		\includegraphics[width=\linewidth]{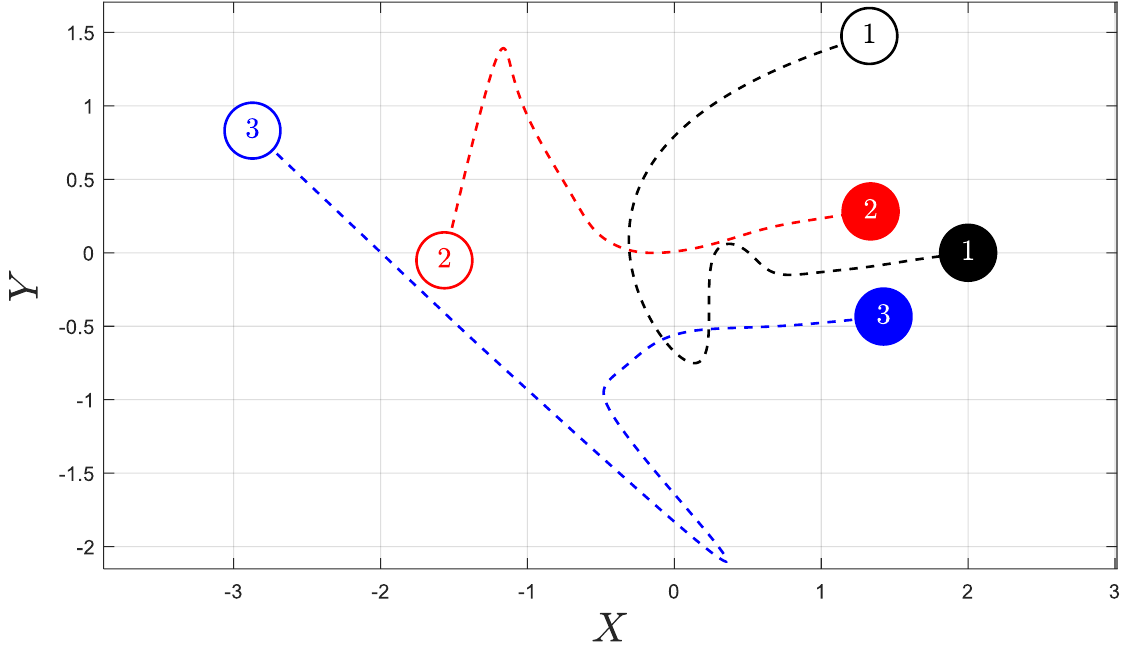}
	\end{subfigure}
	~~~~~
	\begin{subfigure}[t]{0.48\linewidth}
		\centering
		\includegraphics[width=\linewidth]{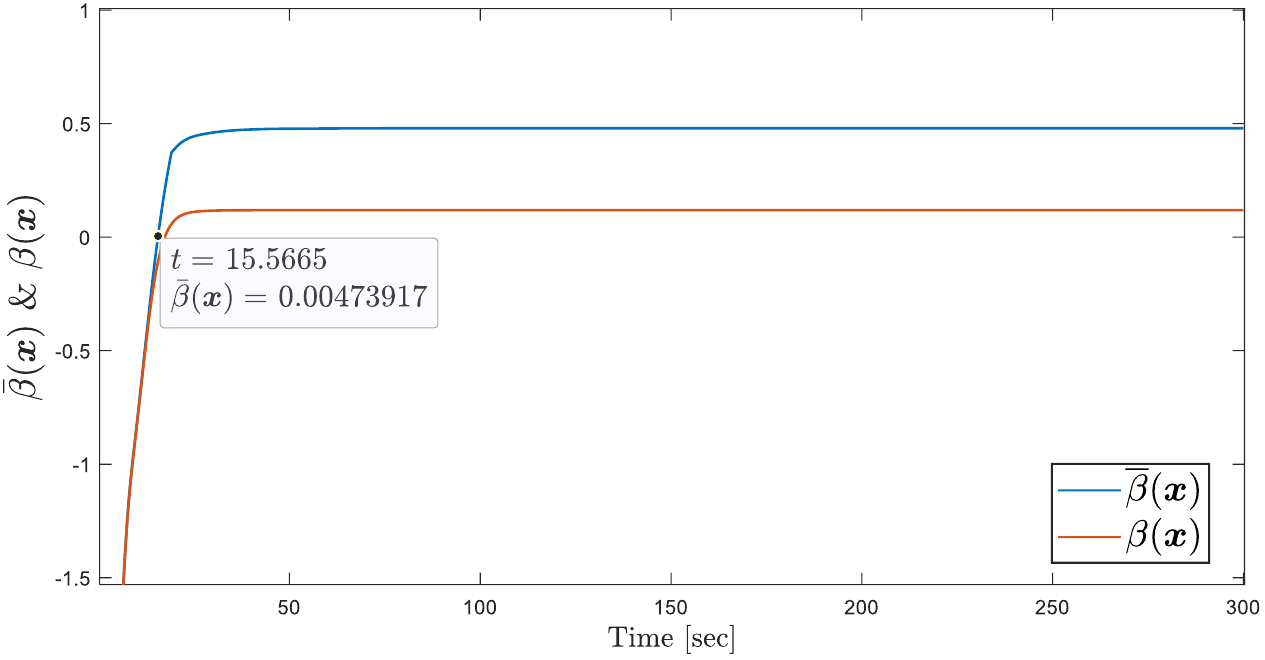}
	\end{subfigure}
	\caption{Simulation results under multi-agent spatial constraints of Case E.\vspace{-0.2cm}}
	\label{fig:simu_exam5_rendezvous}
\end{figure}

%%%%%%%%%%%%%%%%%%%%%%%%%%%%%%%%%%%%%%%%%%%%%%%%%%%%%%%%%%%%%%%%%%%%%%%%%%%%%%%%   
\section{Conclusion}
\label{sec:conclu}
In this work, we addressed long-term spatial constraints in multi-agent systems, where agents not only satisfy their own constraints but also assist others by forming a desired configuration collaboratively. We first formulated the problem as a centralized optimization, introducing an objective function whose positive values indicate constraint satisfaction, with higher values signifying better fulfillment. To design distributed control protocols, we derived an alternative objective function expressed as the sum of local functions, each dependent only on the agent's own constraints, This enabled us to propose a distributed optimization scheme approximating the centralized solution. We also established conditions for the convexity and strict convexity of the global objective function. Finally, using a continuous-time distributed optimization algorithm, we developed a control protocol for single integrator agents. We expect that extending our method will help tackle emerging challenges in multi-robot systems with spatial constraints, including collaborative coordination under spatiotemporal specifications and formation control.

%%%%%%%%%%%%%%%%%%%%%%%%%%%%%%%%%%%%%%%%%%%%%%%%%%%%%%%%%%%%%%%%%%%%%%%%%%%%%%%%
\bibliographystyle{ieeetr}
\bibliography{Refs}
%%%%%%%%%%%%%%%%%%%%%%%%%%%%%%%%%%%%%%%%%%%%%%%%%%%%%%%%%%%%%%%%%%%%%%%%%%%%%%%%

\end{document}